\documentclass[reqno]{amsart}
\usepackage{lineno,hyperref}
\hypersetup{
  colorlinks   = true, 
  urlcolor     = blue, 
  linkcolor    = blue, 
  citecolor   = red 
}
\usepackage{amsmath,amsfonts,amssymb,amsthm,mathrsfs}
\usepackage{enumerate}
\usepackage[sort&compress,numbers]{natbib}
\usepackage[inline]{enumitem}

\newtheorem{corollary}{Corollary}

\newtheorem{lemma}{Lemma}
\newtheorem{proposition}{Proposition}
\newtheorem{remark}{Remark}
\numberwithin{equation}{section}

\newcommand{\naturals}{\mathbb{N}}
\newcommand{\reals}{\mathbb{R}}
\newcommand{\complexes}{\mathbb{C}}
\newcommand{\integers}{\mathbb{Z}}

\newcommand{\spec}[1]{\operatorname{spec}{#1}}
\newcommand{\iprod}[2]{\langle #1, #2 \rangle}

\newcommand{\tnorm}[1]{\| #1 \|_{1}}

\newcommand{\tr}[1]{\operatorname{tr}{#1}}
\newcommand{\ptr}[2]{\operatorname{tr}_{#1}{#2}}
\newcommand{\proj}[1]{\mathsf{P}_{#1}}

\newcommand{\comm}[2]{[#1,#2]}
\newcommand{\acomm}[2]{\{#1,#2\}}

\newcommand{\hadj}{*}

\newcommand{\id}{\mathrm{id}}

\newcommand{\tclass}[1]{B_1 (#1)}
\newcommand{\ket}[1]{| #1 \rangle}
\newcommand{\bra}[1]{\langle #1 |}

\newcommand{\hilbert}{\mathscr{H}}
\newcommand{\hilberte}{\mathscr{H}_{\mathrm{e}}}
\newcommand{\hilbertse}{\mathscr{H}_{\mathrm{se}}}
\newcommand{\hamint}{H_{\mathrm{int.}}}
\newcommand{\matr}[1]{M_{#1}(\complexes)}
\newcommand{\matrd}{\matr{d}}
\newcommand{\envstate}{\rho_{\mathrm{e}}}
\newcommand{\bzero}{\mathcal{B}_{0}}

\renewcommand{\vec}[1]{\boldsymbol{#1}}

\begin{document}
\title{Markovian dynamics under weak periodic coupling}
\author{Krzysztof Szczygielski}
\address{Institute of Theoretical Physics and Astrophysics, Faculty of Mathematics, Physics\newline%
\indent  and Informatics, University of Gda\'{n}sk, 80-308 Gda\'{n}sk, Poland}%
\email{krzysztof.szczygielski@ug.edu.pl}%
\date{\today}

\begin{abstract}
We examine a completely positive and trace preserving evolution of finite dimensional open quantum system, coupled to large environment via periodically modulated interaction Hamiltonian. We derive a corresponding Markovian Master Equation under usual assumption of weak coupling using the projection operator techniques, in two opposite regimes of very small and very large modulation frequency. Special attention is granted to the case of uniformly (globally) modulated interaction, where some general results concerning the Floquet normal form of a solution and its asymptotic stability are also addressed.
\end{abstract}

\maketitle

\section{Introduction}

In recent years, certain advancements were made in the field of theory of open quantum systems governed by periodically modulated Hamiltonians. In particular, general theory of Markovian evolution was established in \cite{Alicki2006b,Szczygielski2014} and later led to quite diverse spectrum of various results. These included purely mathematical ones, like the open systems incarnation of Howland's time independent formalism \cite{Szczygielski2020} and Floquet description of periodic Lindbladians in commutative setting \cite{Szczygielski_2021}, as well as purely physical outcomes, ranging from periodic dynamical decoupling scheme in Markovian regime to Lindblad treatment of quantum heat engines and photovoltaic cells \cite{Szczygielski2013,Szczygielski2015,Gelbwaser-Klimovsky2015,Alicki2015a,Alicki2017,Alicki2017a,Alicki2018a,Alicki2019}. So far, the time-periodicity in the microscopic model of open system has been considered systematically in regard to the system's self-Hamiltonian, while the interaction term and the Hamiltonian of the environment usually remained time-independent (we note however, that certain advancements were made in general case of time-dependent interaction models, for example for piecewise-constant Hamiltonians as in \cite{Merkli_2008}). In this article, we reformulate existing approach by relocating a functional time dependence of the model from system's Hamiltonian to the interaction term. We present a systematic derivation of Markovian Master Equation governing the system's reduced density operator under usual assumption of weak coupling between system and its environment. We limit our analysis to \emph{periodic} interaction Hamiltonian, which is a mathematically well-tractable case. We believe that proposed formalism of periodic coupling mechanism could potentially find applications in some subareas within broadly understood quantum theory, like quantum computations and error correction, control theory and thermodynamics (possibly in theory of quantum heat engines).
\vskip\baselineskip
The main results of our analysis are presented in sections \ref{sec:WPC} and \ref{sec:UPS}. In section \ref{sec:WPC}, characterizing a general idea of \emph{weak periodic coupling}, we discuss a microscopic (Hamiltonian) model of finite-dimensional system S, which is coupled to external, infinite environment E (reservoir) via bounded interaction Hamiltonian of general form $\lambda A(t)$, where $A(t)$ is periodic with period $T>0$ and self-adjoint, while $\lambda > 0$ is a small (compared to relevant energy scale) dimensionless coupling parameter. First, we derive an integral version of quantum Nakajima-Zwanzig equation (in section \ref{sec:TheModel}) describing evolution of a compound system S+E, using a projection operator approach. Next, we consider a reduced evolution of subsystem S in \emph{weak coupling limit} $\lambda\to 0^+$ in two opposite cases of both very large (in section \ref{sec:WCL}) and very small (section \ref{sec:WCLA}) frequency $\Omega = 2\pi / T$ of interaction Hamiltonian. In the former case, we employ a traditional approach of \emph{weak coupling limit} by Davies \cite{Davies1974}, while the \emph{adiabatic limit} regime by Davies and Spohn \cite{Davies1978} is used in the latter case. Main results of this part of a paper are presented in propositions \ref{prop:FinalWCL} and \ref{prop:Afinal}, respectively, where it is shown that in either case the reduced dynamical maps describing subsystem S are completely positive and trace preserving, and are subject to Markovian Master Equations. Section \ref{sec:UPS} is devoted to a special, simplified case of \emph{uniform periodic steering} where the interaction Hamiltonian is assumed in a form $\lambda g(t) A$ for $g(t)$ real and periodic, $A$ self-adjoint and bounded, and $\lambda$ again being small. It is shown that subsystem S can be then described by a \emph{commutative} Lindbladian family and thus, the induced reduced dynamics admits a product structure (Floquet normal form) which can be calculated exactly (proposition \ref{prop:CLF}) after applying results from \cite{Szczygielski_2021}. This section is then concluded with some comments concerning general algebraic properties of the solution (in proposition \ref{prop:FloquetNormalForm}), as well as its asymptotic stability (section \ref{sec:Asymptotics}).
\vskip\baselineskip
The notation will be mostly traditional. For Hilbert space $\mathscr{X}$, the algebra of bounded linear maps on $\mathscr{X}$ will be traditionally denoted by $B(\mathscr{X})$ and its Banach subspace of \emph{trace class operators} will be $\tclass{\mathscr{X}}$. Hermitian adjoint of any operator $A$ will be denoted by $A^\hadj$. Identity element in algebra $\mathscr{A}$ will be denoted by $I_\mathscr{A}$, while identity map over a given linear space (clear from the context) will be simply $\id{}$. Partial trace of operator $A$ with respect to $\mathscr{X}$ will be $\ptr{\mathcal{X}}{A}$. Occasionally, we will use a dot symbol to indicate differentiation with respect to time variable.

\section{Weak periodic coupling}
\label{sec:WPC}

\subsection{Reduced dynamics}

We start with sketching a general microscopic model of weak periodic coupling. All the results of this section are obtained by application of the projection operator techniques employed by e.g. Nakajima \cite{Nakajima1958}, Zwanzig \cite{Zwanzig1960}, Davies \cite{Davies1974,Davies1976} and Spohn \cite{Davies1978} and we will largely accept notation used therein.

\subsubsection{The model}
\label{sec:TheModel}

Let us consider an open quantum system S, described by finite-dimensional Hilbert space $\hilbert \simeq \complexes^d$ , a constant Hamiltonian $H$ and algebra of observables $B(\hilbert)\simeq \matrd$. We introduce a following spectral decomposition of $H$ (including multiplicities),
\begin{equation}\label{eq:HspecDecomp}
	H = \sum_{k=1}^{d} \epsilon_k P_k, \quad P_k = \ket{\varphi_k}\bra{\varphi_k}, \quad \iprod{\varphi_k}{\varphi_l} = \delta_{kl},
\end{equation}
as well as a set of \emph{Bohr frequencies} $\{\omega = \epsilon_k - \epsilon_l\}$ of $H$, being at the same time a spectrum of corresponding derivation $\comm{H}{\cdot\,}$ defined as a commutator on $\matrd$.

System S is coupled to \emph{environment} E, described by its own Hilbert space $\hilberte$, Hamiltonian $H_{\mathrm{e}}\in B(\hilberte)$ and a constant density operator $\envstate$. In order to validate the \emph{Markovian approximation}, which we invoke eventually, we assume $\hilberte$ infinite-dimensional. The system S+E, treated as a whole, is then described by Hilbert space $\hilbertse = \hilbert \otimes \hilberte$ and algebra $B(\hilbertse)$.

Coupling between subsystems S and E will be realized by a self-adjoint, bounded time-periodic \emph{interaction Hamiltonian} $\hamint(t)$, which we put in general form as
\begin{equation}\label{eq:HintPeriodic}
	\hamint(t) = \sum_{\mu} g_{\mu}(t) \, S_\mu \otimes R_\mu ,
\end{equation}
where $S_\mu \in \matrd$, $R_\mu \in B(\hilberte)$ and $\{g_\mu\}$ is a finite family of complex, piecewise-continuous \emph{steering functions}, periodic with period $T$. Later on, we will simplify \eqref{eq:HintPeriodic} by putting just one steering function, so the \emph{uniform periodic steering} will take place. With no loss of generality \cite{Rivas2012}, one can choose environment operators $R_{\mu}$ appearing in \eqref{eq:HintPeriodic} to be of vanishing expectation value,
\begin{equation}\label{RmuExpVal0}
	\tr{R_\mu \envstate} = 0.
\end{equation}
The reduced density operator of S is contained inside Banach space $\matrd_{1} = (\matrd , \tnorm{\cdot})$ for $\tnorm{\cdot}$ being the trace norm, which we identify isometrically with a closed subspace
\begin{equation}
	\bzero=\matrd_{1} \otimes\envstate
\end{equation}
of $B(\hilbertse)$. We also define two projection operators $\proj{0}$, $\proj{1}$ by setting, for any $a\in B(\hilbertse)$,
\begin{equation}
	\proj{0}(a) = (\ptr{\hilberte}{a}) \otimes \envstate, \quad \proj{1} = \id - \proj{0}
\end{equation}
such that $B(\hilbertse) = \bzero \oplus \mathcal{B}_1$, for $\mathcal{B}_i = \proj{i}B(\hilbertse)$, $i \in \{ 0,\,1 \}$, being the subspaces describing systems S and E, respectively.

The whole Hamiltonian of S+E is periodic and can be decomposed into a sum of two bounded parts
\begin{equation}
	H_{\mathrm{se}}(t) = H_{\mathrm{f.}}+\lambda \hamint(t),
\end{equation}
where $H_{\mathrm{f.}} = H \otimes I_{\hilberte} + I_{\hilbert} \otimes H_{\mathrm{e}}$ is the free part and $\lambda > 0$ is a small coupling parameter. We also introduce two derivations on $B(\hilbertse)$ associated with Hamiltonians $H_{\mathrm{f.}}$ and $\hamint(t)$,
\begin{equation}\label{eq:DerivationsDefinition}
	Z = -i \comm{H_{\mathrm{f.}}}{\cdot\,}, \quad A_t = -i \comm{\hamint(t)}{\cdot\,}.\
\end{equation}

\subsubsection{Reduced evolution and Nakajima-Zwanzig equation}

As system S+E is considered closed, its joint density operator $v^{\lambda}_{t} \in \tclass{\hilbertse}$ undergoes a \emph{reversible evolution} defined by a strongly continuous and differentiable family of completely positive and trace preserving maps $\{V^{\lambda}_{t} : t\in\reals_+\}$ such that $v_t = V_t (v_0)$ for some positive semi-definite trace class operator $v_0 \in \tclass{\hilbertse}$. Family $\{V_t\}$ is subject to von Neumann equation
\begin{equation}\label{eq:vonNeumann}
	\frac{dV^{\lambda}_{t}}{dt} = (Z+\lambda A_t)V_{t}^{\lambda},
\end{equation}
being our starting point. Note, that we emphasized dependence of the solution on coupling parameter $\lambda$. Similarly to the time-independent case \cite{Davies1974,Davies1978} we introduce notation $A_{t}^{ij} = \proj{i} A_t \proj{j}$ for $i,j \in \{0,1\}$.

\begin{lemma}\label{lemma:AtProperties}
We have  $A_{t}^{00} = 0$ and $A_t = A_{t}^{10} + A_{t}^{01} + A_{t}^{11}$.
\end{lemma}

\begin{proof}
Take any $a\in B(\hilbertse)$ and let $\rho = \ptr{\hilberte}{a}$, so $\proj{0}(x) = \rho \otimes \envstate$; calculating explicitly, we have
\begin{equation}
	A_{t}^{00}(a) = \sum_{\mu} g_{\mu}(t) \proj{0} (S_\mu \rho \otimes R_\mu \envstate ) = \sum_{\mu} g_{\mu}(t) (\tr{R_\mu \envstate}) S_\mu \rho = 0
\end{equation}
by \eqref{RmuExpVal0}. Second claim follows directly by putting $A_t = (\proj{0}+\proj{1}) A_t (\proj{0}+\proj{1})$.
\end{proof}
Lemma \ref{lemma:AtProperties} allows to slightly rewrite von Neumann equation,
\begin{equation}
	\frac{dV^{\lambda}_{t}}{dt} = (Z+\lambda A_{t}^{11} + \lambda \Delta_t)V_{t}^{\lambda}, \quad V_{0}^{\lambda} = \id{},
\end{equation}
where $\Delta_t = A_{t}^{10} + A_{t}^{01}$. Treating $\Delta_t$ as a perturbation, one obtains, by usual techniques \cite{Kato1966,Davies1978}, a general expression for \emph{principal fundamental solution} $V_{t}^{\lambda}$ in terms of integral equation
\begin{equation}\label{eq:Vlambda}
	V^{\lambda}_{t} = U^{\lambda}_{t} + \lambda \int\limits_{0}^{t} U_{t,t'}^{\lambda} \Delta_{t'} V_{t'}^{\lambda} \, dt' ,
\end{equation}
where $U^{\lambda}_{t}$, $t\in\reals_+$ is a solution of unperturbed ODE
\begin{equation}\label{eq:vonNeumannUnperturbed}
	\frac{dU^{\lambda}_{t}}{dt} = (Z+\lambda A_{t}^{11})U^{\lambda}_{t}, \quad U^{\lambda}_{0} = \id{}
\end{equation}
and $U^{\lambda}_{t,s} = U_{t}^{\lambda}(U^{\lambda}_{s})^{-1}$, $s \in [0,t]$, is its corresponding \emph{state transition matrix} (propagator). Similarily, $U_{t,s}^{\lambda}$ may also be re-expressed in terms of integral equation by the same approach,
\begin{equation}\label{eq:UtLambda}
	U_{t,s}^{\lambda} = e^{(t-s)Z} + \lambda \int\limits_{s}^{t} e^{(t-t')Z} A_{t'}^{11} U_{t',s}^{\lambda} \, dt' ,
\end{equation}
for uniformly continuous semigroup $\{e^{tZ} : t\in\reals_+\}$ being a fundamental solution of equation $\dot{\varphi}_t = Z(\varphi_t)$. One can check, by easy computation, that the following simple lemma holds:

\begin{lemma}\label{lemma:AcommZ}
We have the following:
\begin{enumerate}
	\item $\comm{\proj{i}}{Z} = 0$, $\comm{\proj{i}}{U^{\lambda}_{t}} = 0$ for $i \in \{0,1\}$, i.e.~map $U_{t}^{\lambda}$ leaves subspaces $\bzero$, $\mathcal{B}_1$ invariant;
	\item Restriction of $U_{t}^{\lambda}$ to $\bzero$ is an isometry;
	\item $\proj{0} U_{t,s}^{\lambda} = \proj{0} e^{(t-s)Z}$.
\end{enumerate}
\end{lemma}

Since the projection $\proj{0}$ is defined simply as a partial trace with respect to environment degrees of freedom, projecting $v_{t}^{\lambda}$ produces a reduced density operator of subsystem S (tensorized with constant state of E). Following \cite{Davies1974}, let us define a time-dependent linear epimorphism $W_{t}^{\lambda} : \tclass{\hilbertse} \to \bzero$ by setting
\begin{equation}
	W_{t}^{\lambda} = \proj{0} V_{t}^{\lambda} \proj{0}, \quad t\in \reals_+
\end{equation}
which gives rise to the \emph{reduced denisty operator} $\rho^{\lambda}_{t}$ subject to equality
\begin{equation}
	\rho^{\lambda}_{t} \otimes \envstate = W_{t}^{\lambda} (v_0) = \Lambda_{t}^{\lambda} (\rho_{0}) \otimes \envstate,
\end{equation}
where $\Lambda_{t}^{\lambda} (\rho_{0}) = \ptr{\hilberte}{W^{\lambda}_{t}(v_{0})}$ and $v_{0} = \rho_{0} \otimes \envstate \in \bzero$ is an initial factorized state of compound system. The mapping $t \mapsto \Lambda_{t}^{\lambda}$ defines the notion of celebrated \emph{quantum dynamical map}, i.e.~a completely positive, trace norm contraction on $\matrd_1$. By strong differentiability assumption of $V_{t}^{\lambda}$ as a map on $\tclass{\hilbertse}$ and continuity of $\proj{0}$, function $t\mapsto\Lambda_{t}^{\lambda}$ is clearly also strongly continuous and differentiable on $\matrd_1$. The seemingly artificial dependence of the solution on parameter $\lambda$ will then be removed via the usual limiting procedure $\lambda \to 0^+$.

\begin{proposition}\label{prop:Wt}
Operator-valued function $t \to W_{t}^{\lambda}$, $t\in\reals_+$, satisfies the Nakajima-Zwanzig integral equation
\begin{equation}\label{eq:Wt}
	W^{\lambda}_{t} = \proj{0} e^{tZ}+ \lambda^2 \int\limits_{0}^{t} dt' \int\limits_{0}^{t'} e^{(t-t')Z} A_{t'}^{01} U^{\lambda}_{t',t''} A_{t''}^{10} W^{\lambda}_{t''} \, dt'' .
\end{equation}
\end{proposition}

\begin{proof}
We just sketch the calculations as they are straightforward. Recursively substitute formula \eqref{eq:Vlambda} in place of $U_{t'}^{\lambda}$ under the integral to obtain $\mathcal{O}(\lambda^2)$ expansion; then, multiply the expression from both sides by $\proj{0}$ and apply lemma \ref{lemma:AcommZ} and complementarity of projections, $\proj{i}\proj{j} = \delta_{ij}\proj{j}$, in order to arrive at \eqref{eq:Wt} after some manipulations.
\end{proof}

In following paragraphs we will investigate closer the reduced dynamics of subsystem S under assumption of asymptotically weak interaction. In particular, two opposing regimes of either very large or very small driving frequency $\Omega$ (compared to system's relevant Bohr frequencies) will be assumed in order to obtain mathematically sound expressions for completely positive dynamics of the reduced system. The first case of fast driving will be examined with application of traditional weak coupling limit procedure due to Davies \cite{Davies1974,Davies1976}. On the other hand, the \emph{adiabatic limit} approach by Davies and Spohn \cite{Davies1978} will then be shown to be well-suited for description of the remaining second case of asymptotically small driving frequency.

\subsection{Weak coupling limit: large driving frequency}
\label{sec:WCL}

In this section, we present a formal \emph{weak coupling limit} procedure following Davies \cite{Davies1974,Davies1976} by considering convergence of certain nets of Volterra integral operators in formal limit $\lambda \to 0^+$. For convenience, we first switch to interaction picture generated by free Hamiltonian $H_{\mathrm{f.}}$, i.e.~we impose an isomorphism
\begin{equation}\label{eq:WtildeDef}
	\tilde{W}_{t}^{\lambda} = e^{-tZ} W_{t}^{\lambda}
\end{equation}
and next introduce a \emph{rescaled time}
\begin{equation}
	\tau = \lambda^2 t,
\end{equation}
which, however, we will be still conventionally denoting by $t$. This, together with \eqref{eq:WtildeDef} after substituting to \eqref{eq:Wt} and changing the integration order leads to expression
\begin{equation}
	\tilde{W}_{t}^{\lambda} = \proj{0}+ \lambda^2 \int\limits_{0}^{\lambda^{-2}t} dt'' \int\limits_{t''}^{\lambda^{-2}t} e^{-t' Z} A_{t'}^{01} U^{\lambda}_{t',t''} A_{t''}^{10} e^{t'' Z }\tilde{W}^{\lambda}_{t''} \, dt' .
\end{equation}
This is further re-written by substitutions $\lambda^2 t'' = s$ and next $x=t' - \lambda^{-2}s$ as
\begin{equation}\label{eq:WtauTilde}
	\tilde{W}_{t}^{\lambda} = \proj{0} + \int\limits_{0}^{t} e^{-\lambda^{-2}s Z} K^{\lambda}_{t-s,s} e^{\lambda^{-2}s Z} \tilde{W}^{\lambda}_{s} \, ds ,
\end{equation}
where  $K^{\lambda}_{t,s}$, the \emph{memory kernel}, is
\begin{equation}\label{KlambdaTauSigma}
			K^{\lambda}_{t,s} = \int\limits_{0}^{\lambda^{-2}t} e^{-xZ} A_{x+\lambda^{-2}s}^{01} U^{\lambda}_{x+\lambda^{-2}s , \lambda^{-2}s} A_{\lambda^{-2}s}^{10} \, dx .
\end{equation}
 We also introduce few additional notions. Notice, that action of derivation $Z$, when restricted to subspace $\bzero$, may by represented as
\begin{equation}
	Z\proj{0}(a) = -i \comm{H}{\rho}\otimes\envstate, \quad \rho = \ptr{\hilberte}{a},
\end{equation}
for any $\rho\in B(\hilbertse)$, where $H$ was the system's part of free Hamiltonian. Applying spectral decomposition \eqref{eq:HspecDecomp} of $H$ and evaluating the commutator, one quickly checks that
\begin{equation}
	Z\proj{0}(a) = -i \sum_{kl} (\epsilon_k-\epsilon_l) P_k \rho P_l \otimes \envstate = -i\sum_{\omega} \omega (\mathcal{E}_\omega \otimes I ) (\rho \otimes \envstate),
\end{equation}
for $\mathcal{E}_\omega (\rho) = \sum_{(k,l)\,\sim\,\omega}P_k \rho P_l$, where the summation in taken only over such pairs of indices $(k,l)$, that $\epsilon_k-\epsilon_l = \omega$. This in turn leads to following spectral decompositions
\begin{equation}\label{eq:Zdecomposition}
	\proj{0}Z = Z\proj{0} = -i \sum_{\omega} \omega Q_\omega, \quad \proj{0}e^{tZ} = \sum_\omega e^{-i\omega t} Q_\omega ,
\end{equation}
where operators $Q_\omega = \mathcal{E}_\omega \otimes I$ project onto different subspaces, $Q_\omega Q_{\omega'} = \delta_{\omega\omega'}Q_{\omega'}$. These allow to define, for any linear map $X$ on $\bzero$,
\begin{equation}\label{eq:TimeAveraging}
	X^{\sharp} = \sum_{\omega} Q_\omega X Q_\omega = \lim_{t\to\infty}\frac{1}{2t}\int\limits_{-t}^{t} e^{-is\comm{H_{\mathrm{f.}}}{\cdot}} X e^{is\comm{H_{\mathrm{f.}}}{\cdot}} ds,
\end{equation}
where the last equality, i.e.~the \emph{time-averaging} \cite{Davies1974}, can be directly shown using \eqref{eq:Zdecomposition}.

For any periodic function $t\mapsto X_t \in B(\bzero)$ we define its \emph{Fourier series} via usual formulas
\begin{equation}
	X_t \sim \sum_{n\in\integers} \hat{X}_{n} e^{in\Omega t}, \quad \hat{X}_{n} = \frac{1}{T}\int\limits_0^T X_t e^{-in\Omega t} dt ,
\end{equation}
regardless of problem of its formal convergence. Accordingly, for any periodic function $f : [0,T) \to \complexes$ we will denote its Fourier series by $\sum_{n\in\integers} \hat{f}(n) e^{in\Omega t}$ for coefficients $\hat{f}(n) = \frac{1}{T}\int_0^T f(t) e^{-in\Omega t} \, dt$. Then, one can in particular express function $t \mapsto A_t$ via its Fourier expansion
\begin{equation}\label{eq:FourierAt}
	A_t \sim \sum_{n\in\integers} \hat{A}_n e^{in\Omega t}, \quad \hat{A}_n = -i \sum_{\mu} \hat{g}_{\mu}(n) \comm{S_\mu \otimes R_\mu}{\,\cdot\,} .
\end{equation}
For technical reasons, we introduce two assumptions: first, we assume that the set of Bohr frequencies is $\Omega$\emph{-congruence free}, namely that no two distinct frequencies $\omega$, $\omega'$ exist such that $\omega - \omega' = k\Omega$ for any $k \in \integers \setminus \{0\}$; this assumption can be met if, for example, the driving frequency $\Omega$ is very large, i.e.~comparable with (or greater than) all Bohr frequencies of the system. Second, we assume that each steering function $g_\mu$ is \emph{bounded}, \emph{piecewise continuous} and of \emph{piecewise-uniformly convergent Fourier series}; the last assumption reads explicitly, that for each $g_\mu$ there exists an increasing sequence $(\tau_{j}^{\mu})_{j\in\naturals} \subset [0,T)$, such that each restriction $\left. g_\mu\right|_{(\tau_j, \tau_{j+1})}$ admits \emph{uniformly convergent Fourier series},
\begin{equation}\label{eq:UniformlyFourier}
	\lim_{N\to\infty} \sup_{t\in (\tau_j, \tau_{j+1})}{\left| g_{\mu}(t) - \sum_{n=-N}^{N} \hat{g}_{\mu}(n) e^{in\Omega t}\right|} = 0.
\end{equation}
This condition is guaranteed if, for example, functions $g_\mu$ are chosen to be piecewise differentiable. This assumption has two important implications: first, $g_\mu (t)$ is representable as its Fourier series at almost every (a.e.) point in $[0,T)$, apart from discrete set $\{\tau_{j}^{\mu}\}$ where the series converges rather to the mean of left-sided and right-sided limits of $g_\mu (t)$ (if, for instance, function is discontinuous at this point). Second, boundedness implies $g_\mu \in L^2 ([0,T))$ and so
\begin{equation}
	\sum_{n\in\integers}|\hat{g}_\mu (n)|^2 < \infty .
\end{equation}
This assumption allows to propose a following straightforward lemma on convergence of Fourier series of map $A_t$:

\begin{lemma}
There exists a strictly increasing sequence $(\theta_j)_{j\in\naturals}\subset [0,T)$ such that Fourier series $\sum_{n\in\integers}\hat{A}_{n} e^{in\Omega t}$ of map $A_t$ converges uniformly (w.r.t. operator norm) on each open interval $(\theta_j, \theta_{j+1})$ and pointwise a.e. in $[0,T)$.
\end{lemma}

\begin{proof}
For all the following, denote partial Fourier sums of $A_t$ and $g_\mu (t)$ by $S_N (A_t)$ and $S_N (g_\mu (t))$, respectively.

The existence of sequence $(\theta_j)$ is simply proved by partially ordering the set $\mathcal{D} = \bigcup_{\mu}\{ \tau_{j}^{\mu} \}_{j\in\naturals}$ of all divergence points, i.e.~for each $\theta_j$ there will be at least one $\mu$ such that $g_\mu (\theta_j)$ is not a pointwise limit of a corresponding Fourier series. Let us denote
\begin{equation}
	\mathcal{C}_{\mu} = [0,T) \setminus \{ \tau_{j}^{\mu}\}_{j\in\naturals}.
\end{equation}
Since $\theta_j$ are all divergence points, intervals $\Theta_j = (\theta_j, \theta_{j+1})$ cover a whole $[0,T)$ as densely as possible. Take any $\theta \in \Theta_j$. Then, there exists some $\mu$, $\mu'$ and $k$, $k'$ such that $\theta_j = \tau_{k}^{\mu}$, $\theta_{j+1} = \tau_{k'}^{\mu'}$ and $\tau_{k}^{\mu} < \theta < \tau_{k'}^{\mu'}$. Clearly,
\begin{equation}
	\tau_{k}^{\mu} = \max_{x < \theta}\{x \in \mathcal{D}\}, \quad \tau_{k'}^{\mu'} = \min_{x > \theta}\{x \in \mathcal{D}\},
\end{equation}
so $\theta \in \mathcal{C}_\mu$ for each $\mu$ and $	\Theta_j \subset \bigcap_{\mu} \mathcal{C}_\mu$ for all $j$. One estimates
\begin{equation}\label{eq:AtEstim}
	\sup_{t\in \Theta_j} \left\| A_t - S_N (A_t) \right\| \leqslant \sup_{t\in \Theta_j} \sum_\mu \left| S_N (g_\mu(t)) - g_\mu (t) \right| \| S_\mu \otimes R_\mu \|
\end{equation}
for $\| \cdot \|$ denoting operator norm over $B(\hilbertse)$. If Fourier series of any function converges uniformly over some interval $\mathcal{I}$, then naturally it must also converge uniformly over any sub-interval $\mathcal{J} \subseteq \mathcal{I}$. Since each $\Theta_j$ is a sub-interval of all sets $\mathcal{D}_\mu$, Fourier series of every function $g_\mu$ must converge uniformly over every $\Theta_j$, and so the upper bound in \eqref{eq:AtEstim} converges to 0 when $N\to\infty$. Hence, Fourier series of $A_t$ converges uniformly over every interval $(\theta_j,\theta_{j+1})$. The pointwise convergence is then immediate.
\end{proof}

Let $\mathscr{V} = \mathcal{C} ([0,t_\ast], \bzero)$ be a Banach space of continuous, $\bzero$-valued functions on compact interval $[0,t_\ast]$, $t_\ast \in \reals_+$, complete with respect to supremum norm
\begin{equation}
	\| \varphi \|_{\mathscr{V}} = \sup_{t \in [0,t_\ast]} \| \varphi (t) \|_{\bzero}.
\end{equation}
We next introduce a following bounded operators acting on $\bzero$,
\begin{subequations}
	\begin{equation}\label{eq:KlambdaSigma}
		\tilde{K}_{s}^{\lambda} = \int\limits_{0}^{\infty} \proj{0} e^{-xZ} A_{x+\lambda^{-2}s} e^{xZ} A_{\lambda^{-2}s} \proj{0} \, dx ,
	\end{equation}
	\begin{equation}\label{eq:Kn}
		K_n = \int\limits_{0}^{\infty} \proj{0} e^{-xZ}\hat{A}_{n} e^{xZ} \hat{A}_{-n} \proj{0} \, dx, \quad n\in\integers .
	\end{equation}
\end{subequations}
After easy manipulations, one can express operators $K_n$ in more direct form as
\begin{equation}\label{eq:KnExplicit}
	K_n = \sum_{\mu\nu\omega} \hat{g}_{\mu}(n) \hat{g}_{\nu}(-n) \proj{0} \comm{S_{\mu\omega}\otimes\hat{R}_{\mu}(\omega)}{\comm{S_\nu \otimes R_\nu}{\,\cdot\,}}\proj{0},
\end{equation}
where we introduced the one-sided Fourier transforms of reservoir operators in interaction picture,
\begin{equation}\label{eq:RmuFourier}
	\hat{R}_\mu (\omega) = \int\limits_{0}^{\infty} e^{-i\omega t} \left(e^{itH_{\mathrm{e}}} R_\mu e^{-itH_{\mathrm{e}}}\right) \, dt,
\end{equation}
as well as operators $S_{\mu\omega}$ defined by
\begin{equation}\label{eq:Smuomega}
	S_{\mu\omega} = \sum_{(k,l)\,\sim\,\omega} P_k S_\mu P_l,
\end{equation}
where orthogonal projections $P_k$ were used to define a spectral decomposition of $H$ \eqref{eq:HspecDecomp}. Following the original construction \cite{Davies1974}, we also define three bounded Volterra-type integral operators on $\mathscr{V}$,
\begin{subequations}
	\begin{equation}\label{eq:VolterraHlambda}
		\mathcal{H}_\lambda (\varphi)(t) = \int\limits_{0}^{t} e^{-\lambda^{-2}s Z} K^{\lambda}_{t-s,s} e^{\lambda^{-2}s Z} (\varphi(s)) \, ds ,
	\end{equation}
	\begin{equation}\label{eq:HtildeOfPhi}
		\tilde{\mathcal{H}}_{\lambda}(\varphi)(t) = \int\limits_{0}^{t} e^{-\lambda^{-2}s Z} \tilde{K}^{\lambda}_{s} e^{\lambda^{-2}s Z} (\varphi(s)) \, ds ,
	\end{equation}
	\begin{equation}\label{eq:KofVarphi}
		\mathcal{K}(\varphi)(t) = \int\limits_{0}^{t} \sum_{n\in\integers} K_{n}^{\sharp} (\varphi(s)) \, ds,
	\end{equation}
\end{subequations}
where convergence of infinite operator series in \eqref{eq:KofVarphi} comes via proposition \ref{prop:KseriesConvergence} (see below). The key point in the construction is to show that the two nets $(\mathcal{H}_\lambda)_\lambda$, $(\tilde{\mathcal{H}}_\lambda)_\lambda$, $\lambda \in\reals_+$ become arbitrarily close to operator $\mathcal{K}$ (in the strong operator topology sense) when $\lambda$ is close to 0 and hence, they approximate an exact solution of Nakajima-Zwanzig equation arbitrarily well. In result, the exact dynamics turns out to be actually described by \emph{quantum dynamical semigroup}, i.e.~the regime of weak coupling limit is stable against fast fluctuations of interaction term.

\begin{proposition}\label{prop:KseriesConvergence}
If all functions $g_\mu \in L^2 ([0,T))$, then $\sum_{n\in\integers}K_n$ converges uniformly in $B(\bzero)$.
\end{proposition}

\begin{proof}
Let $S_N = \sum_{n=-N}^{N} K_n$ denote a partial sum of $\sum_{n\in\integers}K_n$. Take $N,M > 0$ and assume $M>N$ with no loss of generality. Using \eqref{eq:KnExplicit} one can estimate
\begin{equation}\label{eq:SnSm}
	\| S_N - S_M \| \leqslant 4 \sum_{\mu\nu\omega} |\sigma_N - \sigma_M| \| S_{\mu\omega}\otimes \hat{R}_{\mu}(\omega) \| \| S_\nu \otimes R_\nu \|,
\end{equation}
where $\sigma_N = \sum_{n=-N}^{N} \hat{g}_{\mu}(n)\hat{g}_{\nu}(-n)$. This series however converges, which we show now: let us denote by $\vec{\hat{g}}_{\mu} = (\hat{g}_{\mu}(n))_{n\in\integers}$ a sequence of Fourier coefficients of function $g_\mu$. By properly restricting the enumeration index $n$, we split $\vec{\hat{g}}_\mu$ into three subsequences, $\vec{\hat{g}}_{\mu} = \vec{\hat{g}}_{\mu}^{-}\oplus\hat{g}_{\mu}(0)\oplus\vec{\hat{g}}_{\mu}^{+}$, where $\vec{\hat{g}}_{\mu}^{-} = (\hat{g}_{\mu}(n))_{n<0}$ and $\vec{\hat{g}}_{\mu}^{+} = (\hat{g}_{\mu}(n))_{n>0}$ are simply the ``left-hand side'' and ``right-hand side'' parts of $\vec{\hat{g}}_{\mu}$. Since $g_\mu \in L^2 ([0,T))$, we obviously have $\vec{\hat{g}}_{\mu} \in l^2 (\integers)$, and since the 2-norm is easily seen to satisfy
\begin{equation}
	\| \vec{\hat{g}}_{\mu} \|_{l^2 (\integers)}^{2} = |\hat{g}_{\mu}(0)|^2 + \| \vec{\hat{g}}_{\mu}^{+} \|_{l^2}^{2} + \| \vec{\hat{g}}_{\mu}^{+} \|_{l^2}^{2},
\end{equation}
we have $\vec{\hat{g}}_{\mu} \in l^2 (\integers)$ iff $\vec{\hat{g}}_{\mu}^{\pm}\in l^2$. The H\"{o}lder's inequality then allows to estimate, after some easy algebra,
\begin{equation}\label{eq:Holder}
	\sum_{n\in\integers} |\hat{g}_{\mu}(n)\hat{g}_{\nu}(-n)| \leqslant |\hat{g}_{\mu}(0)\hat{g}_{\nu}(0)| + \| \vec{\hat{g}}_{\mu}^{+} \|_{l^2} \| \vec{\hat{g}}_{\nu}^{-} \|_{l^2} + \| \vec{\hat{g}}_{\mu}^{-} \|_{l^2} \| \vec{\hat{g}}_{\nu}^{+} \|_{l^2}
\end{equation}
which is finite. Therefore, series $\sum_{n\in\integers} \hat{g}_{\mu}(n)\hat{g}_{\nu}(-n)$ converges (absolutely), so sequence $(\sigma_N)$ of its partial sums is Cauchy, i.e.~$|\sigma_N-\sigma_M|\to 0$ as $N,M\to\infty$. Hence, the upper bound in \eqref{eq:SnSm} also converges to 0 and $(S_N)$ is a Cauchy sequence in $B(\bzero)$.
\end{proof}

\begin{proposition}
Operator $\tilde{K}_{s}^{\lambda}$ admits an explicit Fourier series expansion
\begin{equation}\label{eq:KlambdaSigmaFourier}
	-\sum_{n,m\in\integers} e^{i(n+m)\Omega \lambda^{-2}s} \hat{g}_{\mu}(n)\hat{g}_{\nu}(m) \proj{0} \comm{S_{\mu\omega}\otimes\hat{R}_{\nu}(\omega+n\Omega)}{\comm{S_\nu \otimes R_\nu}{\,\cdot\,}}\proj{0}
\end{equation}
converging uniformly (w.r.t. operator norm in $B(\bzero)$ and pointwise a.e. in $[0,T)$.
\end{proposition}

\begin{proof}
Validity of formula \eqref{eq:KlambdaSigmaFourier} can be checked by simple algebra. First, one checks that operators $S_\mu$ satisfy
\begin{equation}\label{eq:SmuIPexpansion}
	e^{it\comm{H}{\cdot\,}}(S_\mu) = e^{iHt}S_\mu e^{-iHt} = \sum_{\omega} S_{\mu\omega} e^{i\omega t},
\end{equation}
which follows from \eqref{eq:Smuomega} and Baker--Campbell--Haussdorff formula. After Fourier expanding operators $A_t$ in expression \eqref{eq:KlambdaSigma} for $\tilde{K}_{s}^{\lambda}$, inserting \eqref{eq:SmuIPexpansion} and putting $\hat{R}_{\mu}(\omega + n\Omega)$ according to \eqref{eq:RmuFourier}, one arrives at \eqref{eq:KlambdaSigmaFourier}. Convergence of the resulting series follows directly from assumed convergence of appropriate Fourier expansions for functions $g_\mu$.
\end{proof}

At this point, it is necessary to introduce the time dependent \emph{autocorrelation functions of the environment} by setting
\begin{equation}
	f_{\mu\nu}(t) = \tr{\left[\tilde{R}_{\mu}(t)^{\hadj} R_{\nu} \envstate \right]}, \quad \tilde{R}_{\mu}(t) = e^{itH_{\mathrm{e}}} R_{\mu} e^{-itH_{\mathrm{e}}}.
\end{equation}
The following lemma will be of some use:

\begin{lemma}\label{lemma:XiInL1}
Let $x\mapsto \Phi_{y}(x) \in B(\bzero)$ be given by
\begin{equation}
	\Phi_{y}(x)(a) = \proj{0} e^{-xZ} A_{x+y} e^{xZ} A_{y} \proj{0} (a)
\end{equation}
for $a\in\bzero$ and some $y\in\reals_+$, and denote $\xi_y (x) = \| \Phi_{y}(x) \|$ for $\| \cdot \|$ being the operator norm in $B(\bzero)$. Then, if $f_{\mu\nu}\in L^1 ((0,\infty))$ for all $(\mu, \nu)$, then also $\xi_y \in L^1 ((0,\infty))$ for all $y\in\reals_+$.
\end{lemma}

\begin{proof}
Applying definition \eqref{eq:DerivationsDefinition} of derivation $A_t$ and expanding a resulting double commutator one obtains, for any $\rho \otimes \envstate \in \bzero$,
\begin{align}
	\Phi_{y}(\rho \otimes \envstate) &= \sum_{\mu\nu}g_{\mu}(y)\overline{g_{\nu}(x+y)} f_{\nu\mu}(x)\left(\tilde{S}_{\nu}(x)^{\hadj} S_\mu \rho - S_\mu \rho \tilde{S}_{\nu}(x)^{\hadj}\right) \\
	&+ \sum_{\mu\nu}\overline{g_{\mu}(y)}g_{\nu}(x+y) \overline{f_{\nu\mu}(x)} \left( \rho S_{\mu}^{\hadj} \tilde{S}_{\nu}(x) - \tilde{S}_{\nu}(x) \rho S_{\mu}^{\hadj} \right) \nonumber ,
\end{align}
where $\tilde{S}_{\mu}(x) = e^{ixH} S_\mu e^{-ixH}$, which allows to estimate
\begin{equation}
	\xi_y (x) \leqslant 4 \sum_{\mu\nu} |g_{\mu}(y)g_{\nu}(x+y)| \| S_\mu \| \| S_\nu \| |f_{\nu\mu}(x)| ,
\end{equation}
which, along with boundedness of functions $g_\mu$, leads to
\begin{equation}
	\int\limits_{0}^{\infty} |\xi_y (x)| dx \leqslant 4 \sum_{\mu\nu} \| g_\mu \|_{L^\infty} \| g_\nu \|_{L^\infty} \| S_\mu \| \| S_\nu \| \int\limits_{0}^{\infty} |f_{\nu\mu}(x)| dx
\end{equation}
(for $\| \cdot \|_{L^\infty}$ denoting supremum norm) which is finite by integrability of autocorrelation functions; hence, $\xi_y$ is also integrable on $(0,\infty )$.
\end{proof}

\begin{proposition}\label{proposition:KlambdaConv}
Assume that all reservoir autocorrelation functions $f_{\mu\nu}$ satisfy
\begin{equation}\label{eq:AutocorrFunConditionEpsilon}
	\int\limits_{0}^{\infty} |f_{\mu\nu}(t)| (1+t)^{\epsilon} dt < \infty
\end{equation}
for some $\epsilon > 0$. Then, for any $s,t \in \reals_+$, $s \leqslant t$, net $(K^{\lambda}_{t-s,s}-\tilde{K}_{s}^{\lambda})_\lambda$ converges to 0 uniformly in $B(\bzero)$ as $\lambda\to 0^+$.
\end{proposition}

\begin{proof}
By recursive substitutions in formula \eqref{eq:UtLambda}, we re-express the propagator in terms of power series
\begin{subequations}\label{eq:UlambdaPowerSeries}
	\begin{equation}
		U^{\lambda}_{x+\lambda^{-2}s , \lambda^{-2}s} = e^{xZ} + \sum_{n=1}^{\infty} \lambda^n b_{n}(x,s),
	\end{equation}
	\begin{equation}
		b_{n}(x,s ) = \int\limits_{0}^{x} dt_1 \, ... \, \int\limits_{0}^{t_{n-1}} dt_n\, e^{xZ} \tilde{A}_{t_1 + \lambda^{-2}s}^{11}(t_1) \, ... \, \tilde{A}_{t_n + \lambda^{-2}s}^{11}(t_n)
	\end{equation}
\end{subequations}
for $\tilde{A}_{t}^{11} (s) = e^{-sZ} A^{11}_{t} e^{sZ}$. Then, one quickly estimates
\begin{align}\label{eq:KlambdaEst}
	\| K^{\lambda}_{t-s,s}-\tilde{K}_{s}^{\lambda} \| &\leqslant \int\limits_{\lambda^{-2}(t-s)}^{\infty} \| \proj{0} e^{-xZ} A_{x+\lambda^{-2}s} e^{xZ} A_{\lambda^{-2}s} \proj{0} \| \, dx \\
	&+ \sum_{n=1}^{\infty} \lambda^n  \| a_{n}^{\lambda} (t-s) \|, \nonumber
\end{align}
for $a_{n}^{\lambda} (t)$ given by formula
\begin{equation}\label{eq:aN}
	a_{n}^{\lambda} (t) = \int\limits_{0}^{\lambda^{-2}t} dt_0 \, ... \, \int\limits_{0}^{t_{n-1}} dt_n \, \proj{0} \tilde{A}_{t_0 + \lambda^{-2}s}(t_0) \left[\prod_{j=0}^{n}\tilde{A}_{t_j +\lambda^{-2}s}^{11}(t_j)\right] A_{\lambda^{-2}s}\proj{0}.
\end{equation}
The integral at the r.h.s.~of \eqref{eq:KlambdaEst} vanishes as $\lambda \to 0^+$ due to lemma \ref{lemma:XiInL1}. The remaining power series may be then shown to vanish term by term whenever $\lambda\to 0^+$. We do not present a detailed proof of this claim, as it is virtually the same as the proof provided in \cite[Theorems 2.3, 3.4 and 3.5]{Davies1974} with the only difference coming from time dependence of functions $g_\mu$; this however does not pose a problem, neither does alter the general proof guidelines, since all functions $g_\mu$ are bounded and one simply refines the upper bounds of all operator-valued functions by the supremum norm of functions $g_\mu$. In consequence, $\| K^{\lambda}_{t-s,s}-\tilde{K}_{s}^{\lambda}\|\to 0$ regardless of $t$, $s$.
\end{proof}

\begin{proposition}\label{prop:VolterraConvergence}
Let $t_\ast \in \reals_+$ and let again $\mathscr{V} = \mathcal{C} ([0,t_\ast], \bzero)$ (complete with supremum norm). Then, as $\lambda\to 0^+$,
\begin{enumerate}
	\item\label{claim:VolterraConvergence1} net $(\mathcal{H}_\lambda - \tilde{\mathcal{H}}_\lambda)_\lambda$ converges strongly to 0 on $\mathscr{V}$;
	\item\label{claim:VolterraConvergence2} nets $(\mathcal{H}_\lambda)_\lambda$ and $(\tilde{\mathcal{H}}_\lambda)_\lambda$ converge strongly to $\mathcal{K}$ on $\mathscr{V}$.
\end{enumerate}
\end{proposition}

\begin{proof}
Let $\varphi \in \mathscr{V}$. Claim (\ref{claim:VolterraConvergence1}) follows from estimation
\begin{equation}
	\| \mathcal{H}_\lambda (\varphi ) - \tilde{\mathcal{H}}_\lambda (\varphi) \|_{\mathscr{V}} \leqslant \sup_{t\in [0,t_\ast]} \int\limits_{0}^{t} \| K^{\lambda}_{t-s,s}-\tilde{K}^{\lambda}_{s} \| \cdot \| \varphi(s ) \| \, ds.
\end{equation}
By proposition \ref{proposition:KlambdaConv}, the integrand converges to 0 with $\lambda\to 0^+$, and so $(\mathcal{H}_\lambda - \tilde{\mathcal{H}}_\lambda)_\lambda$ converges to 0 pointwise (strongly). For claim (\ref{claim:VolterraConvergence2}) it suffices to show $(\tilde{\mathcal{H}}_\lambda)_\lambda \to \mathcal{K}$ strongly. Decomposing $e^{\pm \lambda^{-2}s Z}$ similarly to \eqref{eq:Zdecomposition} and using Fourier decomposition \eqref{eq:FourierAt} of $A_t$, we put \eqref{eq:HtildeOfPhi} for any $\varphi \in \bzero$ as
\begin{align}\label{eq:HtildeLambdaExplicit}
	\tilde{\mathcal{H}}_\lambda (\varphi)(t ) &= \int\limits_{0}^{t} ds \int\limits_{0}^{\infty} e^{-\lambda^{-2}s Z} \proj{0} e^{-xZ} A_{x+\frac{s}{\lambda^2}} e^{xZ} A_{\frac{s}{\lambda^2}} e^{\lambda^{-2}s Z} \proj{0} (\varphi(s)) \, dx \\ 
	&= \sum_{\omega\omega'}\sum_{n,m\in\integers} \int\limits_{0}^{t} e^{i(\omega_{n} -\omega'_{m})\lambda^{-2}s} Q_\omega W_{mn} Q_{\omega'} (\varphi (s)) \, ds \nonumber ,
\end{align}
where $\omega_{n} = \omega + n\Omega$ was used as a shorthand for shifted Bohr frequencies and
\begin{equation}
	W_{mn} = \int\limits_{0}^{\infty}\proj{0} e^{-xZ} \hat{A}_m e^{xZ}\hat{A}_{-n} \proj{0} \, dx.
\end{equation}
Notice, that we have $\sum_{m}\delta_{mn} W_{mn} = K_n$. This, together with \eqref{eq:HtildeLambdaExplicit} yields
\begin{equation}\label{eq:HminusKupperBound}
	\| \tilde{\mathcal{H}}_\lambda (\varphi ) - \mathcal{K}(\varphi) \|_{\mathscr{V}} \leqslant \sup_{t\in [0,t_\ast]}{\left\| \sum_{\omega\omega'}\sum_{m,n\in\integers} Q_\omega W_{mn} Q_{\omega'}\left( \zeta^{\lambda}_{\omega\omega' mn}(t)\right) \right\|},
\end{equation}
where
\begin{equation}\label{eq:Zeta}
	\zeta_{\omega\omega' mn}^{\lambda}(t) = \int\limits_{0}^{t} \left(e^{i(\omega_n-\omega'_{m})\lambda^{-2}s}-\delta_{\omega\omega'}\delta_{mn}\right) \varphi(s ) \, ds .
\end{equation}
By \emph{Riemann-Lebesgue theorem}, 
\begin{equation}\label{eq:RL}
	\lim_{\lambda\to 0^+}\int\limits_{0}^{t}e^{i(\omega_n-\omega'_{m})\lambda^{-2}s}\varphi(s )\, ds = \delta_{\omega_n \omega'_m}\int\limits_{0}^{t}\varphi(s)\, ds
\end{equation}
uniformly over $[0,t_\ast]$. In presence of our earlier assumption of $\Omega$-congruence freedom of Bohr frequencies, we have $\omega_n = \omega'_m$ iff $\omega=\omega'$ and $n=m$, i.e.~$\delta_{\omega_n \omega'_m}=\delta_{\omega\omega'}\delta_{mn}$ and the upper bound in \eqref{eq:HminusKupperBound} converges to 0 uniformly for any $\varphi \in \mathscr{V}$ and hence, the strong convergence of $(\tilde{\mathcal{H}}_\lambda)_\lambda$ to $\mathcal{K}$ is clear.
\end{proof}

\begin{remark}\label{remark:WCLapplicability}
In presence of the direct application of Riemann-Lebesgue theorem, it is important to demand $\Omega$ to be \emph{large}, i.e.~at least comparable with typical Bohr frequencies $\omega$, in order to ensure that the weak coupling limit result represents a well enough approximation of the exact dynamics. By direct calculation one can in fact show, assuming $\varphi (s)$ is continuously differentiable, that
\begin{equation}
	\int\limits_{0}^{t}e^{i\Delta\lambda^{-2}s}\varphi(s )\, ds = \mathcal{O}\left( \frac{\lambda^2}{\Delta} \right)
\end{equation}
for $\Delta = \omega_n -\omega'_{m} \neq 0$; this naturally yields the integral converges to 0 in the limit $\lambda\to 0^+$. However, the larger $|\Delta|$ becomes, the closer the integral is to 0 on complex plane for fixed $\lambda > 0$, since the function $f(\lambda) = \lambda^2/\Delta$ rescales and becomes much more ``stretched'' horizontally. But this also means that $f(\lambda)$ changes more rapidly for lower values of $\Delta$ while shifting $\lambda$ towards 0. This means, that if $|\Delta|$ is small, then the error introduced by performing the limit $\lambda\to 0^+$ over map $\tilde{\mathcal{H}}_\lambda$ with respect to exact dynamics (still characterized by some finite, non-zero parameter $\lambda$) could be potentially greater than in case of slowly changing $f(\lambda)$. The weak coupling limit of $\tilde{\mathcal{H}}_\lambda$ can therefore be a worse approximation of exact solution if $|\Delta|$ is small and hence a large $\Omega$ regime should be preferable.
\end{remark}

The above result in fact concludes the formal construction of reduced dynamics in weak coupling limit regime. Below we give a closing remark:

\begin{proposition}\label{prop:FinalWCL}
Let us define a function $t\mapsto W_t \in B(\bzero)$, $t\in [0,t_\ast]$, by
\begin{equation}
	W_{t} = \proj{0}e^{t Z} + \int\limits_{0}^{t} e^{t Z} \tilde{K} e^{-s Z} W_{s}\, ds , \quad \tilde{K} = \sum_{n\in\integers} K_{n}^{\sharp}
\end{equation}
and let $\tilde{W}_{t}^{\lambda}$ be given by \eqref{eq:WtauTilde} nad \eqref{KlambdaTauSigma}. If all functions $g_\mu$ admit \emph{piecewise-uniformly convergent Fourier series} (see the assumptions introduced prior to equation \eqref{eq:UniformlyFourier}) and if all autocorrelation functions satisfy condition \eqref{eq:AutocorrFunConditionEpsilon} for some $\epsilon > 0$, then for every $t_\ast \in \reals_+$ and every initial state $v_0 = \rho_0 \otimes \envstate \in \bzero$, the following claims hold:
\begin{enumerate}
	\item\label{claim:WCLresultLargeOmegaOne} Net $(v_\lambda)_\lambda\subset \mathscr{V}$ of functions given via
	\begin{equation}\label{eq:vLambda}
		v_\lambda (t) = W_{t}^{\lambda}(v_0)
	\end{equation}
	for $W_{t}^{\lambda} = e^{t Z}\tilde{W}_{t}^{\lambda}$, converges uniformly over $[0,t_\ast]$ to function 
	\begin{equation}\label{eq:vK}
		v(t) = W_t (v_0).
	\end{equation}
	In other words, $W_t$ provides a solution of Nakajima-Zwanzig integral equation in the weak coupling limit $\lambda \to 0^+$.
	\item\label{claim:WCLresultLargeOmegaTwo} Family $\{\Lambda_t : t\in\reals_+\}$ of dynamical maps defined by 
	\begin{equation}
		W_t (v_0) = \Lambda_t (\rho_0) \otimes \envstate
	\end{equation}
	is completely positive and trace preserving contraction semigroup.
\end{enumerate}
\end{proposition}

\begin{proof}
The proof follows general guidelines drawn by Davies \cite{Davies1974} and therefore we will not be overly explicit here. Let us introduce an isometry $\mathcal{E}:\mathscr{V}\to\mathscr{V}$ by setting $\mathcal{E}(\varphi)(t) = e^{t Z}(\varphi(t))$, which is a formal switch to the Schroedinger picture from the interaction picture and $\mathcal{E}^{-1}$ is its reversal. Then, by \eqref{eq:vLambda}, \eqref{eq:vK} and finally \eqref{eq:VolterraHlambda} and \eqref{eq:KofVarphi}, the Schroedinger picture solutions $v,v_\lambda \in \mathscr{V}$ satisfy equations
\begin{equation}
	v_\lambda = \mathcal{E}(w) + \mathcal{E}\mathcal{H}_{\lambda}\mathcal{E}^{-1} (v_\lambda), \quad v = \mathcal{E}(w) + \mathcal{E}\mathcal{K}\mathcal{E}^{-1} (v)
\end{equation}
for $w = \mathcal{E}\proj{0}(v_0)$. By recursive substitutions, one easily checks that these can be cast into
\begin{equation}
	v_\lambda = \sum_{n=0}^{\infty} (\mathcal{E}\mathcal{H}_\lambda\mathcal{E}^{-1})^{n} (w), \quad v = \sum_{n=0}^{\infty} (\mathcal{E}\mathcal{K}\mathcal{E}^{-1})^{n} (w),
\end{equation}
which, together with isometry condition $\| \mathcal{E} \| = 1$ allows to estimate
\begin{equation}\label{eq:vLambdaVest}
		\| v_\lambda - v \| \leqslant \sum_{n=0}^{\infty} \| \mathcal{H}_{\lambda}^{n}(w) - \mathcal{K}^{n}(w) \|.
\end{equation}
Both $\mathcal{H}_\lambda$, $\mathcal{K}$ are Volterra operators and so they can be shown to satisfy inequalities $\|\mathcal{H}_{\lambda}^{n}\| \leqslant C^n t_{*}^{n}
/n!$, $\| \mathcal{K}^n\| \leqslant D^n t_{*}^{n}/n!$ for some constants $C,D > 0$ and, since one clearly has
\begin{equation}
	\| \mathcal{H}_{\lambda}^{n}(w) - \mathcal{K}^{n}(w) \| \leqslant (\|\mathcal{H}_{\lambda}^{n}\|+\|\mathcal{K}^n\|) \| w \|,
\end{equation}
the series in \eqref{eq:vLambdaVest} converges for all $\lambda > 0$. However, as $\mathcal{H}_\lambda \to \mathcal{K}$ strongly by proposition \ref{prop:VolterraConvergence}, each term in the series converges to 0 as $\lambda \to 0^+$ and in the result, $v_\lambda \to v$ uniformly; this proves claim (\ref{claim:WCLresultLargeOmegaOne}). As differentiability of $W_t$ is obvious, after elementary calculations we have
\begin{equation}\label{eq:Dwtau}
	\frac{dW_t}{dt} = \left( Z + e^{t Z}\tilde{K}e^{-t Z} \right) W_t .
\end{equation}
By definition \eqref{eq:Kn} of operators $K_n$ and by time-averaging procedure \eqref{eq:TimeAveraging}, the action of $\tilde{K}$ on $v_0 \in \bzero$ can be expressed as
\begin{align}\label{eq:KtildeOverB0}
	\tilde{K}(v_0) &= - \lim_{t\to\infty}\frac{1}{2t}\int\limits_{-t}^{t} ds \int\limits_{0}^{\infty} \proj{0} \comm{e^{(s-x)Z}(\hat{H}_{n})}{\comm{e^{sZ}(\hat{H}_{-n})}{v_0}} \, dx \\
	&= \lim_{t\to\infty}\frac{1}{2t}\int\limits_{-t}^{t} ds \int\limits_{0}^{\infty} \proj{0} \Bigg( -\sum_{n\in\integers} e^{(s-x)Z}(\hat{H}_n)e^{sZ}(\hat{H}_{-n})\,v_0 \nonumber \\
	&+ \sum_{n\in\integers} e^{sZ}(\hat{H}_{-n})\,v_0\, e^{(s-x)Z}(\hat{H}_{n}) + \sum_{n\in\integers} e^{(s-x)Z}(\hat{H}_n)\,v_0 \, e^{sZ}(\hat{H}_{-n}) \nonumber \\
	&- \left.\sum_{n\in\integers} v_0 \, e^{sZ}(\hat{H}_{-n})e^{(s-x)Z}(\hat{H}_{n}) \right). \nonumber
\end{align}
Next, we change the summation index $n$ to $-n$ in first two sums in \eqref{eq:KtildeOverB0} and replace $\hat{H}_{-n}$ with $\hat{H}_{n}^{\hadj} = \sum_{\mu\nu}\overline{\hat{g}_{\mu}(n)}S_{\mu}^{\hadj}\otimes R_{\mu}^{\hadj}$. The remaining computations then follow the canonical textbook path and therefore we only sketch them briefly. First, we employ expansions \eqref{eq:SmuIPexpansion} for expressions of a form $e^{it\comm{H}{\cdot}} (S_\mu)$; second, we perform a limiting procedure, i.e.~the time-averaging, $\lim_{t\to\infty}\frac{1}{2t}\int_{-t}^{t}e^{is(\omega-\omega')}ds = \delta_{\omega\omega'}$, which allows to suppress fast-oscillating terms of a form $e^{is(\omega-\omega')}$. Third, we re-express the one-sided Fourier transforms of reservoir autocorrelation functions as
\begin{equation}\label{eq:OneSidedFourier}
	\int\limits_{0}^{\infty} e^{-i\omega x} f_{\mu\nu}(x) \, dx = \frac{1}{2} h_{\mu\nu}(\omega) + i \zeta_{\mu\nu}(\omega),
\end{equation}
where $h_{\mu\nu}(\omega)$ is the usual Fourier transform of $f_{\mu\nu}(t)$, and $\zeta_{\mu\nu}(\omega)$ is often obtained with appropriate use of Sochozki formulas \cite{Rivas2012}. Finally, after some algebra we arrive at 
\begin{align}\label{eq:LinIP}
	\tilde{K}(\rho\otimes\envstate) &= \Big( -i\comm{\delta H}{\rho} + G(\rho) \Big) \otimes\envstate,
\end{align}
where $\delta H$ and $G$ admit forms
\begin{subequations}
	\begin{equation}
		\delta H = \sum_{\mu\nu\omega}\xi_{\nu\mu}(\omega) S_{\nu\omega}^{\hadj} S_{\mu\omega},
	\end{equation}
	\begin{equation}
		G(\rho) = \sum_{\mu\nu\omega}\gamma_{\nu\mu}(\omega) \left( S_{\mu\omega}\rho S_{\nu\omega}^{\hadj} - \frac{1}{2}\acomm{S_{\nu\omega}^{\hadj}S_{\mu\omega}}{\rho} \right) ,
	\end{equation}
\end{subequations}
for matrices
\begin{equation}
	\gamma_{\nu\mu}(\omega) = \sum_{n\in\integers}\hat{g}_{\mu}(n)\overline{\hat{g}_{\nu}(n)} h_{\nu\mu}(\omega), \quad \xi_{\nu\mu}(\omega) = \sum_{n\in\integers}\hat{g}_{\mu}(n)\overline{\hat{g}_{\nu}(n)} \zeta_{\nu\mu}(\omega).
\end{equation}
By H\"{o}lder's inequality and estimation \eqref{eq:Holder}, all above series converge. Stone's theorem guarantees that matrices $[\gamma_{\mu\nu}(\omega)]_{\mu\nu}$ and $[\xi_{\mu\nu}(\omega)]_{\mu\nu}$ are positive-semidefinite. Matrix $\delta H$, being explicitly Hermitian, is commonly called the \emph{Lamb shift Hamiltonian} and expresses change of system's energy levels due to environment's influence. By positive semi-definiteness of matrix $[\gamma_{\mu\nu}(\omega)]_{\mu\nu}$, formula \eqref{eq:LinIP} defines a generator of completely positive and trace preserving semigroup. Operators $S_{\mu\omega}$ satisfy commutation relation
\begin{equation}\label{eq:ScommRel}
	\comm{H}{S_{\mu\omega}} = \omega S_{\mu\omega},
\end{equation}
which can be verified by \eqref{eq:HspecDecomp} and \eqref{eq:Smuomega}; this in turn leads to conditions
\begin{equation}\label{eq:dHGcommRel}
	\comm{\delta H}{H} = 0, \quad \comm{G}{\comm{H}{\cdot\,}} = 0,
\end{equation}
i.e.~$G$ and $\comm{H}{\cdot\,}$ commute as maps on $\matrd$, which is sometimes referred to as the \emph{covariance property} \cite{Alicki2006a,Breuer2002,Rivas2012}. Then, from \eqref{eq:Dwtau} we have, for $\rho_t = \Lambda_t (\rho_0)$,
\begin{align}
	\frac{dW_t}{dt}(v_0) &= \left( -i\comm{H + e^{-it\comm{H}{\cdot\,}}(\delta H)}{\rho_t} +  e^{-it\comm{H}{\cdot\,}}\, G\, e^{it\comm{H}{\cdot\,}}(\rho_t)\right)\otimes\envstate \\
	&= \Big( -i\comm{H+\delta H}{\rho_t} + G(\rho_t)\Big) \otimes\envstate \nonumber \\
	&= \frac{d\Lambda_t}{dt}(\rho_0)\otimes\envstate \nonumber,
\end{align}
i.e.~$\Lambda_t = e^{t L}$ for $L = -i\comm{H+\delta H}{\cdot\,}+G$ in standard form.
\end{proof}

\begin{corollary}
To summarize, in the regime of large modulation frequency $\Omega$, the reduced density operator $\rho_t$ of subsystem S satisfies, under weak coupling limit, the Markovian Master Equation
\begin{equation}
	\frac{d\rho_t}{dt} = -i \comm{H_\mathrm{eff.}}{\rho_t} + \sum_{\mu\nu\omega}\gamma_{\nu\mu}(\omega) \left( S_{\mu\omega}\rho_t S_{\nu\omega}^{\hadj} - \frac{1}{2}\acomm{S_{\nu\omega}^{\hadj}S_{\mu\omega}}{\rho_t} \right),
\end{equation}
where $H_\mathrm{eff.} = H+\delta H$ is the physical, \emph{effective} Hamiltonian of S, including the Lamb shift correction term.
\end{corollary}

\subsection{Adiabatic limit: small driving frequency}
\label{sec:WCLA}

As was already mentioned (see remark \ref{remark:WCLapplicability}), the weak coupling limit procedure should be expected to produce a well-enough approximation of the exact dynamics in case of rather large driving frequency $\Omega$, i.e.~comparable with Bohr frequencies of the system of interest. A case of small $\Omega$, characterizing much more physically accessible scenario, has to be treated differently. If $\lambda$ is still a small ($0<\lambda \ll 1$) parameter, horizontal rescaling of steering functions $g_\mu (t)$ by defining
\begin{equation}
	\tilde{g}_\mu (t) = g_\mu (\lambda^2 t) = \sum_{n\in\integers} \hat{g}_\mu e^{in(\lambda^2 \Omega)t}
\end{equation}
effectively mirrors the case of \emph{small} driving frequency $\lambda^2\Omega$, i.e.~of \emph{slowly varying} interaction term. By accordingly re-writing the initial value problem \eqref{eq:vonNeumann} as
\begin{equation}\label{eq:vonNeumannA}
	\frac{dV^{\lambda}_{t}}{dt} = (Z+\lambda A_{\lambda^2 t})V_{t}^{\lambda},
\end{equation}
where 
\begin{equation}
	A_{\lambda^2 t}(\rho ) = -i \comm{H_{\mathrm{int.}}(\lambda^2 t)}{\rho}, \quad H_{\mathrm{int.}}(\lambda^2 t) = \sum_{\mu}g_{\mu}(\lambda^2 t) S_\mu \otimes R_\mu ,
\end{equation}
the formal limit $\lambda\to 0^+$ is in such case governed rigorously by the \emph{quantum adiabatic theorem}, introduced by Davies and Spohn already in 1978 \cite{Davies1978}. Validity of this approach will be granted provided that the frequency of the interaction Hamiltonian is small and the steering functions vary slowly, or, if the actual observation time is short. We will follow the guidelines of the original paper, however as a majority of required proofs resembles the weak coupling limit case, we will vastly limit our analysis to necessary steps.

Let us introduce a rescaled time $\tau = \lambda^2 t$; as opposed to the previous regime of weak coupling limit, the adiabatic limit now prefers \emph{small} values of $\tau$, i.e.~short observation times. The new variable yields a new von Neumann equation
\begin{equation}\label{eq:vonNeumannA2}
	\frac{dV^{\lambda}_{t}}{dt} = (\lambda^{-2} Z+\lambda^{-1} A_{t}) V_{t}^{\lambda},
\end{equation}
where again $\tau$ was replaced by symbol $t$. After performing all the computational steps similarly to proposition \ref{prop:Wt}, a new Nakajima-Zwanzig equation is then derived,
\begin{equation}
	W^{\lambda}_{t} = \proj{0} e^{\lambda^{-2}t Z}+ \int\limits_{0}^{t} e^{\lambda^{-2}(t-s)} K^{\lambda}_{t,s} W^{\lambda}_{s} \, ds,
\end{equation}
for operator kernel $K^{\lambda}_{t,s}$ given by
\begin{equation}
	K^{\lambda}_{t,s} = \int\limits_{0}^{\lambda^{-2}(t-s)} e^{-xZ} A^{01}_{s+\lambda^{2}x} U^{\lambda}_{s+\lambda^{2}x,s} A^{10}_{s}\, dx.
\end{equation}

\begin{proposition}
If all environment autocorrelation functions satisfy \eqref{eq:AutocorrFunConditionEpsilon} for some $\epsilon > 0$, then there exists a continuous function $t\mapsto K_t \in B(\bzero)$ such that, if $t\mapsto\hat{W}^{\lambda}_{t}$ is a fundamental solution of ODE of a form
	\begin{equation}
		\frac{d\hat{W}^{\lambda}_{t}}{dt} = (\lambda^{-2} Z+K_t) \hat{W}^{\lambda}_{t}, \quad \hat{W}_0 = \id{},
	\end{equation}
	then for every $t_\ast \in\reals_+$ and every $v_0 \in \bzero$, nets $(v_\lambda)_\lambda,(g_\lambda)_\lambda \subset \mathscr{V}$ given by
	\begin{equation}\label{eq:vgA}
		v_\lambda (t) = W_{t}^{\lambda}(v_0), \quad g_\lambda (t) = \hat{W}^{\lambda}_{t} (v_0),
	\end{equation}
	converge to each other uniformly over $[0,t_\ast]$.
\end{proposition}

\begin{proof}
Both the existence of appropriate map $K_t$ and the above proposition are essentially the core of quantum adiabatic theorem as formulated in \cite{Davies1978}. After noting, that the unitary propagator $U^{\lambda}_{s+\lambda^{2}x,s}$ formally satisfies
\begin{equation}
	\lim_{\lambda\to 0^+} U^{\lambda}_{s+\lambda^{2}x,s} = e^{s Z},
\end{equation}
as a \emph{solution} to differential equation \eqref{eq:vonNeumannUnperturbed}, and since $A_t$ is piecewise continuous, we also have a following asymptotic equality
\begin{equation}\label{eq:KtauA}
	K_s = \lim_{\lambda\to 0^+} K_{t,s}^{\lambda} = \int\limits_{0}^{\infty} e^{-xZ} A_{s}^{01} e^{xZ} A_{s}^{10} \, dx
\end{equation}
satisfied formally. Showing that \eqref{eq:KtauA} is the right choice for $K_t$ involves, similarly to propositions \ref{proposition:KlambdaConv} and \ref{prop:VolterraConvergence}, showing uniform convergence of net $(K^{\lambda}_{t,s}-K_s)$ to 0; this is again achieved by expanding propagator $U^{\lambda}_{s+\lambda^{2}x,s}$ into power series similar to \eqref{eq:UlambdaPowerSeries} and estimating
\begin{align}\label{eq:KminusKA}
	\| K^{\lambda}_{t,s} - K_s\| &\leqslant \int\limits_{0}^{\lambda^{-2}(t-s)} \left\| e^{-xZ} (A^{01}_{s}-A^{01}_{s+\lambda^2 x}) e^{xZ} A^{10}_{s} \right\| \, dx  \\
	&+ \int\limits_{\lambda^{-2}(t-s)}^{\infty} \left\|e^{-xZ} A_{s}^{01} e^{xZ} A_{s}^{10}\right\| \, dx + \sum_{n=1}^{\infty} \lambda^n \| a_{n}^{\lambda}(t-s) \|\nonumber
\end{align}
for $a_{n}^{\lambda} (t)$ given in a manner similar to \eqref{eq:aN}. Employing lemma \ref{lemma:XiInL1} one easily shows that both the integrands at the r.h.s. of \eqref{eq:KminusKA} are $L^1 ((0,\infty))$ and, by continuity of function $t\mapsto A_t$, both integrals disappear in limit $\lambda\to 0^+$. The remaining series can be then also shown to vanish term-wise by the same arguments as in the weak coupling case, and so \eqref{eq:KminusKA} converges to 0. This in turn yields the uniform convergence as claimed.
\end{proof}

Following the original construction, in order to express the approximated dynamics in possibly most accessible way, we introduce a semigroup
\begin{equation}\label{eq:Xexplicit}
	e^{\lambda^{-2}Zt} = \sum_{\omega} e^{-i\lambda^{-2}\omega t} Q_{\omega},
\end{equation}
as well as a function $t\mapsto Y_t$, satisfying
\begin{equation}\label{eq:YdiffEq}
	\frac{dY_t}{dt} = K_{t}^{\sharp}Y_t, \quad Y_0 = \id{},
\end{equation}
where $\{\omega\}$ again stands for a set of Bohr frequencies of system's self-Hamiltonian, $Q_\omega$ are corresponding spectral projection operators of derivation $Z$ and $X^\sharp$ is again defined as a time-averaging procedure given by \eqref{eq:TimeAveraging}. Then, the following approximation theorem applies:

\begin{proposition}\label{prop:Afinal}
Let $(g_\lambda)_\lambda \subset \mathscr{V}$ be given by \eqref{eq:vgA} and let us define a function $t \mapsto F_{t}^{\lambda} \in B(\bzero)$ by setting
\begin{equation}\label{eq:FlambdaDef}
	F_{t}^{\lambda} = Y_t e^{\lambda^{-2}Zt}.
\end{equation}
Then, for every initial state $v_0 \in \bzero$ and every $t_\ast \in \reals_+$, the following claims hold:
\begin{enumerate}
	\item\label{claim:WCLresultSmallOmegaOne} Net $(h_\lambda)_\lambda \subset \mathscr{V}$ of functions defined via
	\begin{equation}
		h_\lambda (t ) = F_{t}^{\lambda} (v_0)
	\end{equation}
	becomes arbitrarily close to $(g_\lambda)$ as $\lambda\to 0^+$, uniformly over $[0,t_\ast]$. In other words, $F_{t}^{\lambda}$ provides an approximate solution of Nakajima-Zwanzig integral equation, suitable in adiabatic limit regime $\lambda\to 0^+$.
	\item\label{claim:WCLresultSmallOmegaTwo} Family $\{\Lambda_t : t\in\reals_+\}$ of dynamical maps defined by 
	\begin{equation}
		F_{t}^{\lambda} (v_0) = \Lambda_t (\rho_0) \otimes \envstate
	\end{equation}
	is completely positive and trace preserving on $\matrd_1$ and satisfies Markovian Master Equation for time-local, periodic Lindbladian in standard form. Moreover, it may be always cast in a form
	\begin{equation}\label{eq:FloquetLambda}
		\Lambda_t = P_t e^{tX},
	\end{equation}
	where $P_t, X$ are linear maps on $\matrd_1$ and $t\mapsto P_t$ is periodic.
	\end{enumerate}
\end{proposition}

\begin{proof}
Claim (\ref{claim:WCLresultSmallOmegaOne}) is just a slightly modified version of similar result in adiabatic framework \cite[Theorem 2]{Davies1978} and is proved in exactly same way, after putting constant projection operators $Q_\omega$ and simplifying accordingly. Proof of claim (\ref{claim:WCLresultSmallOmegaOne}), by its obvious similarity to proposition \ref{prop:FinalWCL}, is virtually identical. First, equation \eqref{eq:YdiffEq} allows to express $Y_t$ as a series
\begin{equation}\label{eq:Yexplicit}
	Y_t = \id{} + \sum_{n=1}^{\infty} \int\limits_{0}^{t} dt_1 \int\limits_{0}^{t_1} dt_2 \, ... \, \int\limits_{0}^{t_{n-1}} K_{t_1}^{\sharp} K_{t_2}^{\sharp} \, ... \, K_{t_n}^{\sharp}\,dt_n.
\end{equation}
By formula \eqref{eq:Xexplicit} and by property of spectral projections $Q_\omega Q_{\omega'} = \delta_{\omega\omega'}Q_{\omega'}$, one quickly shows that $e^{\lambda^{-2}Zt}$ commutes with every element of the series in \eqref{eq:Yexplicit} and hence $\comm{Y_t}{e^{\lambda^{-2}Zt}} = 0$.

A computation similar to \eqref{eq:KtildeOverB0} allows to find an action of time-averaged memory kernel $K_{t}^{\sharp}$ on $\rho\otimes\envstate$,
\begin{align}
	K_{t}^{\sharp}(\rho\otimes\envstate) &= \lim_{a\to\infty} \frac{1}{2a} \int\limits_{-a}^{a} dy \, e^{yZ} \int\limits_{0}^{\infty} \proj{0} e^{-xZ} A_t e^{xZ} A_t \proj{0} e^{-yZ} \, dx \\
	&= \Big( -i \comm{\delta H_t}{\rho} + G_t (\rho) \Big) \otimes \envstate \nonumber,
\end{align}
where time-dependent Lamb shift term $\delta H_t$ and map $G_t$ are given by
\begin{subequations}
	\begin{equation}
		\delta H_t = \sum_{\omega \mu\nu} \xi_{\nu\mu}(\omega,t) S_{\nu\omega}^{\hadj} S_{\mu\omega},
	\end{equation}
	\begin{equation}
		G_t (\rho) = \sum_{\gamma_{\nu\mu}}  \gamma_{\nu\mu}(\omega,t) \left( S_{\mu\omega} \rho S_{\nu\omega}^{\hadj} - \frac{1}{2}\acomm{S_{\nu\omega}^{\hadj}S_{\mu\omega}}{\rho} \right),
	\end{equation}
\end{subequations}
for functions
\begin{equation}
	\xi_{\nu\mu}(\omega,t) = g_{\mu}(t) \overline{g_{\nu}(t)} \zeta_{\nu\mu}(\omega), \quad \gamma_{\nu\mu}(\omega,t) = g_{\mu}(t)\overline{g_{\nu}(t)} h_{\nu\mu}(\omega).
\end{equation}
In the above, $h_{\nu\mu}(\omega)$ and $\zeta_{\nu\mu}(\omega)$ again stand for real and imaginary parts of one-sided Fourier transforms of environment's autocorrelation functions \eqref{eq:OneSidedFourier}. Matrices $S_{\mu\omega}$ are defined identically as in earlier section and satisfy \eqref{eq:SmuIPexpansion}. By direct check, matrices $[\xi_{\mu\nu}(\omega,t)]_{\mu\nu}$ and $[\gamma_{\mu\nu}(\omega,t)]_{\mu\nu}$ are positive semi-definite for all frequencies $\omega$ and all $t\in\reals_+$. Note, that due to periodicity of functions $g_\mu$, function $t\mapsto K_{t}^{\sharp}$ is periodic as well. Hence,
\begin{equation}
	-i \comm{\delta H_t}{\cdot\,} + G_t, \quad t\in[0,T),
\end{equation}
constitutes for periodic Lindbladian in standard (Lindblad-Gorini-Kossakowski-Sudarshan) form. Now, by applying \eqref{eq:ScommRel}, one obtains, similarly as before, commutation relations equivalent to \eqref{eq:dHGcommRel} with $\delta H$ and $G$ replaced by their appropriate time-dependent counterparts, i.e.~the covariance property
\begin{equation}\label{eq:CovPropA}
	\comm{K_{t}^{\sharp}}{Z} = 0
\end{equation}
holds on $\bzero$ in adiabatic regime, leading immediately to $\comm{Y_t}{Z} = 0$ on $\bzero$ due to \eqref{eq:Yexplicit}. This, together with \eqref{eq:YdiffEq}, \eqref{eq:FlambdaDef} and commutativity condition $\comm{Y_t}{e^{\lambda^{-2}Zt}} = 0$ finally allows to compute
\begin{align}
	\frac{dF^{\lambda}_{t}}{dt} &= K_{t}^{\sharp}Y_t e^{\lambda^{-2}Zt} + \lambda^{-2} Y_t Z e^{\lambda^{-2}Zt} = (\lambda^{-2}Z+K_{t}^{\sharp})F_{t}^{\lambda},
\end{align}
which, for $v_0 = \rho_0 \otimes \envstate$ and $\Lambda_t (\rho_0) = \rho_t$, yields
\begin{align}
	\frac{dF^{\lambda}_{t}}{dt}(v_0) &= (\lambda^{-2}Z+K_{t}^{\sharp})F_{t}^{\lambda}(\rho_0 \otimes \envstate) \\
	&= \left( -i\lambda^{-2}\comm{H}{\rho_t} -i\comm{\delta H_t}{\rho_t} + G_t (\rho_t) \right)\otimes\envstate \nonumber \\
	&= \frac{d\Lambda_t}{dt}(\rho_0)\otimes\envstate , \nonumber
\end{align}
i.e.~$\dot{\Lambda}_t = L_t \Lambda_t$ for periodic Lindladian $L_t = -i\comm{\lambda^{-2}H+\delta H_t}{\cdot\,} + G_t$ in standard form. Finally, due to periodicity of $L_t$, decomposition \eqref{eq:FloquetLambda} of $\Lambda_t$ is a direct consequence of celebrated \emph{Floquet theorem} \cite{Chicone2006} and is known as the \emph{Floquet normal form}.
\end{proof}

\begin{corollary}
Again, we summarize by noting, that the Markovian Master Equation satisfied by $\rho_t$ in the regime of small $\Omega$ will be given, under the adiabatic limit approach, by
\begin{equation}
	\frac{d\rho_t}{dt} = -i \comm{H_\mathrm{eff.}(t)}{\rho_t} + \sum_{\mu\nu\omega}\gamma_{\nu\mu}(\omega,t) \left( S_{\mu\omega}\rho_t S_{\nu\omega}^{\hadj} - \frac{1}{2}\acomm{S_{\nu\omega}^{\hadj}S_{\mu\omega}}{\rho_t} \right),
\end{equation}
for periodic effective Hamiltonian $H_\mathrm{eff.}(t)$ and periodic dissipation term.
\end{corollary}

\section{Uniform periodic steering}
\label{sec:UPS}

In this section we present some closer insight into properties of quantum dynamical maps under \emph{uniform} periodic steering scheme, i.e.~in presence of only one steering function $g : [0,T) \mapsto \reals$. Namely, we assume the interaction Hamiltonian is of simpler form
\begin{equation}
	H_{\mathrm{int.}}(t) = g(t) \sum_\mu S_\mu \otimes R_\mu .
\end{equation}
We focus only on the second case of asymptotically small modulation frequency $\Omega$. Proposition \ref{prop:Afinal} allows to put the Markovian Master Equation (under adiabatic limit regime) for subsystem S in a form
\begin{equation}\label{eq:PeriodicUniLind1}
	\frac{d\rho_t}{dt} = L_t (\rho_t) = -i\comm{H_\mathrm{eff.}(t)}{\rho_t} + G_t (\rho_t),
\end{equation}
where $G_t$ is a periodic Lindbladian in standard form,
\begin{equation}\label{eq:PeriodicUniLind2}
	G_t = |g(t)|^2 D, \quad D(\rho)=\sum_{\mu\nu\omega}  h_{\nu\mu}(\omega) \left( S_{\mu\omega} \rho S_{\nu\omega}^{\hadj} - \frac{1}{2}\acomm{S_{\nu\omega}^{\hadj}S_{\mu\omega}}{\rho} \right).
\end{equation}

\subsection{Floquet normal form of dynamical map}

The reason standing behind such severe simplification is the feasibility and accessibility of mathematical description of resulting reduced dynamics. We emphasize that despite the fact that the existence of general, product structure
\begin{equation}\label{eq:FloquetNormalForm}
	\Lambda_t = P_t e^{tX}
\end{equation}
of fundamental solution of the Markovian Master Equation given in proposition \ref{prop:Afinal}, or the so-called \emph{Floquet normal form}, is guaranteed by Floquet theorem \cite{Chicone2006}, not much can be said about both maps $P_t$ and $e^{tX}$ regarding its mutual complete positivity or trace preservation. This restriction can be however entirely lifted by assuming the family $\{L_t : t\in\reals_+\}$ of Lindbladians is \emph{commutative}, i.e.~it satisfies, for all $t,s \in\reals_+$ and all $\rho\in\matrd$, the condition
\begin{equation}\label{eq:CommCond}
	L_t L_s (\rho) = L_s L_t (\rho).
\end{equation}
In such commutative case, one can find explicit expressions for pair $(P_t,e^{tX})$ as solving the Markovian Master Equation (MME) does not require invoking any cumbersome time-ordering procedure or considering convergence of perturbative series expansions, and is simply reduced to calculating
\begin{equation}
	\Lambda_t = \exp{\int\limits_{0}^{t} L_{t'} \, dt'},
\end{equation}
which may be always achieved by considering the isomorphic representation of maps over $\matrd_1$ as matrices of size $d^2$ and then proper exponentiating.

One easily shows that the commutativity condition \eqref{eq:CommCond} indeed applies in the scheme of uniform steering:

\begin{proposition}\label{prop:CLF}
Family $\{L_t : t\in\reals_+\}$ given by \eqref{eq:PeriodicUniLind1} and \eqref{eq:PeriodicUniLind2} is commutative. The Floquet normal form \eqref{eq:FloquetNormalForm} of induced dynamical map is then given by
\begin{subequations}
	\begin{equation}\label{eq:FloquetPt}
		P_t = \exp{\left[ -i\comm{\mathcal{H}_t - \frac{t}{T}\mathcal{H}_T}{\cdot\,} + \left( \Gamma(t) - \frac{t}{T}\Gamma(T) \right) D \right]},
	\end{equation}
	\begin{equation}\label{eq:FloquetX}
		X = -\frac{i}{T}\comm{\mathcal{H}_T}{\cdot\,} + \frac{\Gamma(T)}{T} D,
	\end{equation}
\end{subequations}
where $\Gamma(t)$ and $\mathcal{H}_t$ are given as antiderivatives
\begin{equation}
	\Gamma(t) = \int\limits_{0}^{t} |g(t')|^2 \, dt', \quad \mathcal{H}_t = \int\limits_{0}^{t} H_\mathrm{eff.}(t') \, dt' .
\end{equation}
\end{proposition}

\begin{proof}
Commutativity condition \eqref{eq:CommCond} comes immediately after employing equalities $\comm{\delta H_t}{H} = 0$ and $\comm{G_t}{\comm{H}{\cdot\,}} = 0$. Next, direct differentiation of operator-valued function
\begin{equation}
	\Phi_t = \exp\Big( -i \comm{\mathcal{H}_t}{\cdot\,} + \Gamma(t) D \Big)
\end{equation}
shows that $\Phi_t$ is the actual \emph{solution} of Master Equation, $\dot{\Phi}_t = L_t \Phi_t$, which means that $\Lambda_t = \Phi_t$. Taking the logarithm of $\Lambda_T$, one arrives at \eqref{eq:FloquetX}. Then, remaining formula \eqref{eq:FloquetPt} is obtained by calculating $P_t = \Lambda_t e^{-tX}$. Periodicity of $P_t$ is then straightforward.
\end{proof}

It was recently shown in \cite{Szczygielski_2021} that commutativity condition imposed on $L_t$ brings some interesting results regarding mutual algebraic properties of both maps $P_t$, $e^{tX}$ of Floquet normal form of $\Lambda_t$. In particular, it was shown that $e^{tX}$ \emph{always} constitutes for completely positive and trace preserving contraction semigroup (i.e.~the quantum dynamical semigroup) and is Markovian in commutative setting. Surprisingly however, \emph{global} Markovianity of $P_t$ is forbidden. The following result is then an immediate corollary of these observations and is proven by employing theorems 2 and 3 of \cite{Szczygielski_2021}:

\begin{proposition}\label{prop:FloquetNormalForm}
The Floquet normal form $(P_t, e^{tX})$ satisfies the following:
\begin{enumerate}
	\item Family $\{e^{tX} : t\in\reals_+\}$ is a Markovian contraction semigroup;
	\item $P_t$ is Markovian in some interval $[t_1, t_2]\subset [0,T)$ if and only if
	\begin{equation}
		|g(t)|^2 \geqslant \frac{\Gamma(T)}{T}
	\end{equation}
	for all $t\in [t_1, t_2]$, and is completely positive for some $t\in [0,T)$, if
	\begin{equation}
		\Gamma(t) \geqslant \frac{t}{T}\Gamma(T) ;
	\end{equation}
	\item $P_t$ is globally Markovian only in trivial case of constant function $g(t)$. Otherwise, there exists a non-empty union $\mathcal{N}\subset[0,T)$ of intervals such that $P_t$ is not allowed to be Markovian anywhere in $\mathcal{N}$.
\end{enumerate}
\end{proposition}

\subsection{Asymptotic stability of solutions}
\label{sec:Asymptotics}

Let $\boldsymbol{\Phi}(t)$ stand for a fundamental matrix solution of linear ODE of a form
\begin{equation}
	\dot{\vec{x}}(t) = \boldsymbol{A}(t) \vec{x}(t)
\end{equation}
over $\reals\times\complexes^n$, for $\boldsymbol{A}(t)\in\matr{n}$ being periodic with period $T$. It is well known \cite{Chicone2006} that the long-time behavior of solutions $\vec{x}(t)$ of the ODE can be fully characterized in terms of spectral properties of matrix $\boldsymbol{\Phi}(T)$, i.e.~a fundamental solution evaluated after one period. Such object gives rise to the \emph{stroboscopic} description of evolution in space of all solutions and is commonly called the \emph{monodromy matrix} of a system. Since the Markovian Master Equation \eqref{eq:PeriodicUniLind1} over $\matrd_1$ in question can be, by suitable \emph{vectorization} procedure, translated to an ODE for functions with values in space $\complexes^{d^2}$, the same stability discussion applies to quantum dynamical maps as well, and is characterized by monodromy operator $\Lambda_T = e^{TX}$, by formula \eqref{eq:FloquetNormalForm}. Seeing a correspondence between monodromy operator and asymptotic stability is immediate if one considers solutions \emph{induced} by the eigenequation of monodromy operator
\begin{equation}
	\Lambda_T (\varphi_j) = \lambda_j \varphi_j ,
\end{equation}
for some $\varphi_j \in \matrd$, $\lambda_j\in\complexes$. One can then define a set of solutions
\begin{equation}
	\rho_j (t) = \Lambda_t (\varphi_j ) = e^{\mu_j t}\phi_j (t) ,
\end{equation}
where functions $\phi_j (t) = P_t (\varphi_j )$ are periodic and $\mu_j\in\complexes$ satisfy $\lambda_j = e^{\mu_j T}$ for any $\lambda_j \in \spec{\Lambda_T}$. Numbers $\mu_j$ and $\lambda_j$ are then called the \emph{characteristic exponents} and \emph{characteristic multipliers} of the ODE (note, that characteristic exponents are not uniquely determined due to non-uniqueness of $\log{\Lambda_T}$), respectively. Under additional assumption of diagonalizability of $\Lambda_T$, set $\{\varphi_j\}$ spans $\matrd$ and $\{\rho_j\}$ is a basis in space of all solutions, i.e.~one can expand any solution $\rho_t$ as
\begin{equation}\label{eq:LimitState}
	\rho_t = \sum_{j=1}^{d^2} c_j \rho_j (t) = \sum_{j=1}^{d^2} c_j e^{\mu_j t}\phi_j (t),
\end{equation}
with coefficients $c_j \in\complexes$ determined by initial condition $\rho_0 = \sum_{j} c_j \varphi_j$. A classical result in theory of linear ODEs then states, that all solutions $\rho_j (t)$ fall into one of three categories, depending on their asymptotic behavior as $t\to\infty$. First, if $|\lambda_j| < 1$ (or, if $\operatorname{Re}\mu_j < 0$), the solution $\rho_j (t)$ vanishes exponentially. If, on the opposite $|\lambda_j| > 1$, the solution grows infinitely in norm, or ``blows up'' at large times (which turns out to be impossible; see below). If $\lambda_j = 1$, then $\rho_j (t)$ oscillates periodically. If, finally $|\lambda_j|=1$ and $\lambda_j \neq 1$, then $\rho_j (t)$ undergoes a \emph{phase shift}, $\rho_j (t+T) = e^{i\theta} \rho_j (t)$ for some $\theta \in [0,2\pi )$. Naturally, solutions falling into first and last category ($|\lambda_j|\leqslant 1$) are called \emph{stable}, and \emph{unstable} otherwise.

\begin{proposition}
The following claims hold for fundamental solution $\Lambda_t$ of Markovian Master Equation \eqref{eq:PeriodicUniLind1}:
\begin{enumerate}
	\item All solutions $\rho_j (t)$ are asymptotically stable, i.e.~no multipliers satisfying condition $|\lambda| > 1$ exist;
	\item Dynamical map $\Lambda_t$ admits an asymptotic limit cycle;
	\item If there are no characteristic multipliers satisfying $|\lambda_j|=1$ other than $1$, then $\Lambda_t$ admits a periodic limit cycle.
\end{enumerate}
\end{proposition}

\begin{proof}
The above result comes as a natural implication of both complete positivity and trace preservation of $e^{TX}$ as a map on $\matrd_1$ and is proved in \cite{Szczygielski_2021}. The core observation here is that spectrum of $\Lambda_T$ lays inside unit disc in complex plane and necessarily contains 1; this then allows to conclude on stability. In a present context, a \emph{limit cycle} of dynamical map denotes such a function $t\mapsto\rho_{t}^{\infty}\in\matrd$, $t\in\reals_+$, that for each initial point $\rho_0$ and $t_0 > 0$ large enough, the restriction of solution $\rho_t = \Lambda_t (\rho_0)$ to time interval $[t_0,\infty )$ is arbitrarily close (in uniform sense) to $\rho_{t}^{\infty}$ in space $\mathcal{C}([t_0,\infty),\matrd_1)$ of continuous, matrix-valued functions, i.e.
\begin{equation}
	\lim_{t_0\to\infty} \sup_{t\geqslant t_0} \tnorm{\rho_t - \rho_{t}^{\infty}} = 0 .
\end{equation}
The exact form of $\rho_{t}^{\infty}$ can be quickly deciphered from \eqref{eq:LimitState} by letting $t$ grow infinitely; then, all terms of a form $e^{\mu t}$ for exponents $\mu$ satisfying $\operatorname{Re}{\mu_j} < 0$ vanish and the only remaining terms are such that $|e^{\mu T}| = 1$. Note, that if $e^{\mu T} = 1$, i.e.~we have $\mu \in 2\pi i T^{-1} \integers$, the corresponding solutions are \emph{periodic}, and if $e^{\mu T}$ lays on the unit circle minus point $\{1\}$, then they are \emph{pseudo-periodic}; pseudo-periodicity in this context means that shifting the solution from time $t$ to $t+T$ shifts coefficient $c_j$ by a phase factor $e^{i\operatorname{Im}{\mu_j}t}$. If there are no multipliers on the unit circle other than $1$, then $\rho_t^\infty$ is simply a periodic limit cycle (steady state).
\end{proof}

\section{Conclusions}

We have presented a formal construction of Markovian evolution of finite dimensional open quantum system under bounded, periodic interaction Hamiltonian in weak coupling limit regime. In particular, we have shown that the problem is well-posed in two opposite regimes of very small and very large driving frequency, where the former case is governed by quantum adiabatic theorem and the latter one is developed under additional assumption of $\Omega$-congruence freedom of Bohr frequencies. As a special case, we have considered a uniform periodic steering scenario, where the resulting Lindbladian was shown to constitute for a commutative family and, in result, more properties of the induced Floquet normal form of the solution were revealed. Below, we also briefly discuss potential extensions of this formalism, as well as possible future research directions:

\begin{enumerate}[label=\arabic*., wide, labelwidth=0pt, labelindent=0pt, itemsep=6pt]
	\item \emph{Quasiperiodic interaction Hamiltonian.} It is natural to extend our results onto the case of the interaction Hamiltonian being \emph{quasiperiodic}, i.e.~expressible in a form $H_{\mathrm{int.}}(t) = H_{\mathrm{int.}}(\Omega_1 t , \, ... \, , \, \Omega_r t)$ for some vector of frequencies $(\Omega_i)\in\reals_{+}^{r}$. Recently, developments towards Markovian evolution under quasiperiodicity of system's self Hamiltonian $H(t)$ were already made \cite{Szczygielski2020a} under assumption of rational independence of frequencies $\{\Omega_i\}$. This potentially suggests that generalization of our results onto a case of quasiperiodic coupling is indeed feasible (perhaps with some technical assumptions).
	\item \emph{Intermediate frequency range.} Analysis which we carried out in Sections \ref{sec:WCL} nad \ref{sec:WCLA} covers the opposing ranges of very large and very small driving frequency which is mainly due to employed time rescaling scheme. The question remains on possible intermediate range of frequencies (for example, comparable with Bohr frequencies) with possibility of Non-markovian nature of induced dynamics.
	\item \emph{Infinite dimensional systems.} Generalization of our results onto the infinite dimensional case is possible, however demands for different treatment (since, for example, the time averaging operation may be ill-defined in this case; see \cite{Davies1976a}).
	\item \emph{Coupling to multiple reservoirs.} Coupling to a finite number of reservoirs, all subject to integrability condition \eqref{eq:AutocorrFunConditionEpsilon}, is straightforward by linear structure of interaction term, as long as the driving frequency remains the same among all the couplings (otherwise the interaction is quasiperiodic).
	\item \emph{Relaxing the $\Omega$-congruence freedom assumption.} In the proof of Proposition \ref{prop:VolterraConvergence}, an assumption of $\Omega$-congruence freedom of Bohr frequencies was introduced. This greatly simplifies the analysis, since in such case the only possibility for condition $\omega_n = \omega^{\prime}_{m}$ to be satisfied is when $\omega=\omega'$ and $n=m$. If this condition is relaxed, then it may happen that there exists a pair $(\omega,\omega')$ of frequencies such that $\omega = \omega' + k_0 \Omega$ for some $k_0 \in \integers\setminus\{0\}$. Then, the condition $\omega_n = \omega^{\prime}_{m}$ is satisfied for infinitely many integers $m,n$ such that $k_0 = m-n$ and all resulting formulas complicate in a significant way. However, it would be still interesting to examine Markovianity properties of the dynamics in such case.
	\item \emph{Generalization onto broader class of reservoirs.} A possibly open question arises concerning applicability of the formalism, and especially Propositions \ref{proposition:KlambdaConv} and \ref{prop:FinalWCL}, in cases of more general reservoirs than those satisfying \eqref{eq:AutocorrFunConditionEpsilon}; a purely bosonic reservoir being a model example. In such cases however, one could still argue for applicability if some certain technical conditions are imposed on the autocorrelation functions (for example, being compactly supported or decaying rapidly enough via introduction of \emph{cut-off frequency}).
\end{enumerate}

\section*{Acknowledgments}
The author is indebted to Prof. Robert Alicki for discussions and to the anonymous Reviewer for comments, which led to improvement of the initial version of the manuscript. Support by the National Science Centre, Poland, via grant No. 2016/23/D/ST1/02043 is greatly acknowledged.

\section*{Data availability}
Data sharing is not applicable to this article as no new data were created or analyzed in this study.


\begin{thebibliography}{10}

\bibitem{Alicki2006b}
R.~Alicki, D.~A. Lidar, and P.~Zanardi.
\newblock {Internal consistency of fault-tolerant quantum error correction in
  light of rigorous derivations of the quantum {M}arkovian limit}.
\newblock {\em Phys. Rev. A}, 73(5):052311, 2006.

\bibitem{Szczygielski2014}
K.~Szczygielski.
\newblock {On the application of {F}loquet theorem in development of
  time-dependent {L}indbladians}.
\newblock {\em J. Math. Phys.}, 55(8):083506, 2014.

\bibitem{Szczygielski2020}
K.~Szczygielski and R.~Alicki.
\newblock {On Howland time-independent formulation of CP-divisible quantum
  evolutions}.
\newblock {\em Reviews in Mathematical Physics}, 32:2050021, 2020.

\bibitem{Szczygielski_2021}
Krzysztof Szczygielski.
\newblock {On the Floquet analysis of commutative periodic Lindbladians in
  finite dimension}.
\newblock {\em Linear Algebra Appl.}, 609:176--202, jan 2021.

\bibitem{Szczygielski2013}
K.~Szczygielski, D.~Gelbwaser-Klimovsky, and R.~Alicki.
\newblock {Markovian master equation and thermodynamics of a two-level system
  in a strong laser field}.
\newblock {\em Phys. Rev. E}, 87(012120):012120, 2013.

\bibitem{Szczygielski2015}
K.~Szczygielski and R.~Alicki.
\newblock {Markovian theory of dynamical decoupling by periodic control}.
\newblock {\em Phys. Rev. A}, 92(2):022349, 2015.

\bibitem{Gelbwaser-Klimovsky2015}
D.~Gelbwaser-Klimovsky, K.~Szczygielski, U.~Vogl, A.~Sa\ss{}, R.~Alicki,
  G.~Kurizki, and M.~Weitz.
\newblock {Laser-induced cooling of broadband heat reservoirs}.
\newblock {\em Phys. Rev. A}, 91:023431, 2015.

\bibitem{Alicki2015a}
R.~Alicki, D.~Gelbwaser-Klimovsky, and K.~Szczygielski.
\newblock Solar cell as a self-oscillating heat engine.
\newblock {\em J. Phys. A: Math. Theor.}, 49(1):015002, 2015.

\bibitem{Alicki2017}
R.~Alicki, D.~Gelbwaser-Klimovsky, and A.~Jenkins.
\newblock A thermodynamic cycle for the solar cell.
\newblock {\em Ann. Phys.}, 378:71--87, 2017.

\bibitem{Alicki2017a}
R.~Alicki.
\newblock {From the GKLS Equation to the Theory of Solar and Fuel Cells}.
\newblock {\em Open. Syst. Inf. Dyn.}, 24(03):1740007, 2017.

\bibitem{Alicki2018a}
R.~Alicki and A.~Jenkins.
\newblock Interaction of a quantum field with a rotating heat bath.
\newblock {\em Ann. Phys.}, 395:69--83, 2018.

\bibitem{Alicki2019}
R.~Alicki.
\newblock A quantum open system model of molecular battery charged by excitons.
\newblock {\em J. Chem. Phys.}, 150(21):214110, 2019.

\bibitem{Merkli_2008}
M.~Merkli and S.~Starr.
\newblock {A Resonance Theory for Open Quantum Systems with~Time-Dependent
  Dynamics}.
\newblock {\em J. Stat. Phys.}, 134(5-6):871--898, 2008.

\bibitem{Davies1974}
E.~B. Davies.
\newblock {Markovian master equations}.
\newblock {\em Commun. Math. Phys.}, 39(2):91--110, 1974.

\bibitem{Davies1978}
E.~B. Davies and H.~Spohn.
\newblock {Open quantum systems with time-dependent Hamiltonians and their
  linear response}.
\newblock {\em J. Stat. Phys.}, 19(5):511--523, 1978.

\bibitem{Nakajima1958}
S.~Nakajima.
\newblock {On Quantum Theory of Transport Phenomena}.
\newblock {\em Progr. Theor. Phys.}, 20(6):948--959, 1958.

\bibitem{Zwanzig1960}
R.~Zwanzig.
\newblock {Ensemble Method in the Theory of Irreversibility}.
\newblock {\em J. Chem. Phys.}, 33(5):1338--1341, 1960.

\bibitem{Davies1976}
E.~B. Davies.
\newblock {\em {Q}uantum {T}heory of {O}pen {S}ystems}.
\newblock Academic Press, London, 1976.

\bibitem{Rivas2012}
{\'{A}}.~Rivas and S.~F. Huelga.
\newblock {\em {Open {Q}uantum {S}ystems: {A}n {I}ntroduction}}.
\newblock Springer, Berlin Heidelberg, 2012.

\bibitem{Kato1966}
T.~Kato.
\newblock {\em Perturbation theory for linear operators}.
\newblock Springer Berlin Heidelberg, 1966.

\bibitem{Alicki2006a}
R.~Alicki and K.~Lendi.
\newblock {\em {Quantum {D}ynamical {S}emigroups and {A}pplications}}.
\newblock Springer, Berlin Heidelberg, 2006.

\bibitem{Breuer2002}
H.-P. Breuer and F.~Petruccione.
\newblock {\em {The theory of open quantum systems}}.
\newblock Oxford University Press, New York, 2002.

\bibitem{Chicone2006}
C.~Chicone.
\newblock {\em {Ordinary {D}ifferential {E}quations with {A}pplications}}.
\newblock Springer, New York, 2006.

\bibitem{Szczygielski2020a}
K.~Szczygielski.
\newblock {On the Lyapunov-Perron reducible Markovian Master Equation}.
\newblock arXiv:2012.01877.

\bibitem{Davies1976a}
E.~B. Davies.
\newblock {Markovian master equations. II}.
\newblock {\em Math. Ann.}, 219(2):147--158, jun 1976.

\end{thebibliography}
\end{document}